\documentclass{article}
\usepackage[T1]{fontenc}
\usepackage[latin1]{inputenc}
\usepackage{mathrsfs}
\usepackage{dsfont}
\usepackage{amsmath,amsfonts,amsthm}
\usepackage{geometry}
\usepackage{microtype}
\usepackage{booktabs}
\usepackage{graphicx}
\usepackage{caption}
\usepackage{subcaption}
\usepackage{xcolor}
\usepackage{array}

\theoremstyle{plain}
\newtheorem{lem}{Lemma}
\newtheorem{thm}{Theorem}

\newtheorem{hyp}{Assumption}
\newtheorem{rem}{Remark}
\newtheorem{ex}{Example}
\theoremstyle{definition}

\DeclareMathOperator*{\argmin}{argmin}
\DeclareMathOperator*{\supp}{supp}

\begin{document}

\title{Asymptotic indifference pricing in exponential L\'{e}vy
  models\footnote{The second author would like to thank Jan Kallsen
    for insightful discussions and comments.}}
\author{Cl\'ement M\'enass\'e \\ Laboratoire de Probabilit\'es et Mod\`eles Al\'eatoires,\\ Universit\'e Paris
    Diderot, Paris, France  \\and\\Centre d'Expertise en \'Etudes et
    Mod\'elisations \'Economiques,\\  GDF SUEZ \and Peter
    Tankov\footnote{The research of the second author is partially supported  by the grant of the Government of Russian Federation 14.12.31.0007.} \\ Laboratoire de Probabilit\'es et Mod\`eles Al\'eatoires,\\ Universit\'e Paris
    Diderot, Paris, France\\
    and\\ International Laboratory of Quantitative Finance,\\
  National Research University Higher School of Economics, Moscow,
  Russia.\\ Email: \texttt{tankov@math.univ-paris-diderot.fr}  }
\date{}
\maketitle

\begin{abstract}

Financial markets based on L\'evy processes are typically incomplete
and option prices depend on risk attitudes of individual agents. In
this context, the notion of utility indifference price has gained popularity in the academic circles. Although theoretically
very appealing, this pricing method remains difficult to apply in
practice, due to the high computational cost of solving the nonlinear
partial integro-differential equation associated to the indifference
price. 
In this work, we develop closed form approximations to exponential
utility indifference prices in exponential L\'evy
models. To this end, we first establish a new non-asymptotic
approximation of the indifference price which extends earlier results on small
risk aversion asymptotics of this quantity. Next, we use this formula to derive a
closed-form approximation of the indifference price by treating the L\'evy model as a perturbation of the Black-Scholes model.
This extends the methodology introduced in a recent paper for smooth linear
functionals of L\'evy processes \cite{cdk} to nonlinear and non-smooth
functionals. Our closed formula represents the indifference price as
the linear combination of the Black-Scholes price and correction terms
which depend on the variance, skewness and kurtosis of the underlying
L\'evy process, and the
  derivatives of the
Black-Scholes price. 
As a by-product, we obtain a simple explicit formula for the spread
between the buyer's and the seller's indifference price. This formula
allows to quantify, in a model-independent fashion, how sensitive a given product is to jump risk
in the limit of small jump size.

\end{abstract}
\medskip



\section{Introduction}
The celebrated Black-Scholes model, which uses the geometric Brownian
motion to describe the dynamics of the assets, is a cornerstone of the
modern mathematical finance. However, it fails to reproduce
significant features of empirically observed stock returns and option
prices, such as fat-tailed distribution {and} implied volatility smile. For this reason, various extensions of the Black-Scholes framework have been developed in the literature. One popular approach is to replace the geometric Brownian motion with the exponential of a L\'evy process. 

L\'evy processes allow to quantify
market risk much more precisely, but the option pricing problem for
such processes becomes more involved. Exponential L\'evy models
typically correspond to incomplete financial markets, meaning that the
agents will not necessarily agree on a unique price for a derivative
product. Instead, the price at which a market agent will accept to buy
or sell a given derivative will depend on his / her risk aversion and
preferences. A commonly used pricing paradigm in this context is the
indifference pricing approach \cite{hodges1989optimal}, which states
that a fair price $p$ of a contingent claim $H$ for a market agent
with utility function $U$ and initial wealth $V_0$ is
the one at which the agent is indifferent between entering and not
entering the transaction:
\begin{align}
\max_\vartheta\mathbb E\left[U\left(V_0 + \int_0^T \vartheta_t
    dS_t\right)\right] = \max_\vartheta \mathbb
E\left[U\left(V_0 + {p} + \int_0^T \vartheta_t dS_t {- H}\right)\right],\label{priceimp}
\end{align}
where $S$ denotes the stock price and the maximum is taken over a
suitable set of admissible trading strategies $\vartheta$. 

In this
paper, we focus more specifically on the exponential (constant
absolute risk aversion) utility function $U(x) = -e^{-\alpha x}$,
where $\alpha>0$ is the risk aversion parameter. 
This leads to a more explicit form for the indifference price:
$$
p^\alpha=\frac{1}{\alpha}\log \frac{\min_{\vartheta } \mathbb
E\left[\exp\left(  -\alpha\int_0^T \vartheta_t dS_t
    {+ \alpha H}\right)\right]}{\min_{\vartheta }\mathbb E\left[\exp\left(-\alpha \int_0^T \vartheta_t
    dS_t\right)\right] }.
$$
Additionally, using the so-called \emph{minimal entropy martingale
  measure} (MEMM) denoted by $\mathbb Q^*$ (see equation \eqref{memm}), the exponential
utility indifference price can be expressed {through a single} optimization
problem:
\begin{equation}
p^\alpha = \frac{1}{\alpha} \log \min_\vartheta \mathbb
E^{*}\left[\exp\left(  -\alpha\int_0^T \vartheta_t dS_t
    {+ \alpha H}\right)\right],\label{ipintro}
\end{equation}
where $\mathbb E^{*}$ stands for the expectation under the measure
$\mathbb Q^*$. 

Nevertheless, computing the utility indifference price \eqref{ipintro} of even a simple European option in
an exponential L\'evy model boils down to solving a non-linear
integro-differential equation (see e.g.,
\cite{jaimungal2005pricing,wu2009pricing}), which is a tough numerical
problem. This makes this approach unsuitable in a production
environment of a bank, where prices must usually be evaluated in real
time. For this reason, asymptotic approximations for
the indifference price in incomplete markets {are of great importance}. 

One approach is to
study the asymptotics when the number of contingent claims or,
equivalently, the risk aversion $\alpha$, is small
\cite{exputility,Kramkov2007,Kramkov2006,mania2005,sixauthor,becherer2006bounded}. In
particular, for the  exponential
utility function, it is known
\cite{mania2005,sixauthor,becherer2006bounded} that {under mild assumptions}, as $\alpha$ tends
to $0$, the indifference price $p^\alpha$ converges to $\mathbb
E^{*}[H]$, the expectation of the
option's pay-off computed under the MEMM, and that the optimal strategy converges to
the quadratic hedging strategy under the MEMM. 

However, approximating the indifference price by the expectation under
the MEMM fails to take
into account the nonlinear features of the price. For this reason, in
\cite{exputility}, the authors compute the first-order correction to
the exponential utility indifference price, and show that it is
proportional to the residual risk of the quadratic hedging strategy
under the MEMM:
\begin{align}
p^\alpha = \mathbb E^{*}[H] + \frac{\alpha}{2} \min_\vartheta
\mathbb E^{*}\left[\left(\int_0^T \vartheta_t dS_t - H +
    \mathbb E^{*}[H]\right)^2\right] + o(\alpha),\quad
\alpha \to 0. \label{prass}
\end{align}
These results are obtained under assumptions which is not straightforward to check for stock
price models with jumps (in particular, Assumption 2 in \cite{exputility}). Similar results for general path-dependent
claims on a Brownian filtration are obtained in
\cite{monoyios2013malliavin} using Malliavin calculus techniques.  Our aim in this paper is therefore to obtain
precise approximations for the indifference prices of options in  exponential
L\'evy models. 

First, in Theorem \ref{thm1}, we establish an approximation for
the indifference price of the following form: 
\begin{align}
p = \mathbb E^{*}[H] + \frac{\alpha}{2} \min_\vartheta
\mathbb E^{*}\left[\left(\int_0^T \vartheta_t dS_t - H +
    \mathbb E^{*}[H]\right)^2\right] + \text{Error}(\alpha). \label{prnoass}
\end{align}
Unlike previous
studies, our formula is non-asymptotic, in the sense that we provide
an explicit bound on the error which is valid
for all values of $\alpha$ smaller than a certain positive constant
rather than asymptotically as $\alpha \to 0$. In addition, this
formula is proven under assumptions which are relatively easy to check
in exponential L\'evy models.  The proof of \eqref{prnoass} is
based on an interplay between the primal and the dual formulation of the
indifference pricing problem to obtain an upper and a lower bound on
the price, and is inspired by similar approaches in
\cite{kallsen2013portfolio} and \cite{henderson2002valuation}. 

Next, we use formula \eqref{prnoass} to develop {a closed form approximation to the exponential utility indifference price} in exponential L\'evy models by treating the
L\'evy model as a perturbation of the Black-Scholes model (see Theorem
\ref{thm2}). In view of our
non-asymptotic representation for the indifference price, this boils
down to approximating the expectation of the pay-off under the
MEMM, approximating the residual risk of the quadratic hedging
strategy, and controlling the error term in \eqref{prnoass}. To this end, we use the approach
suggested in a recent paper \cite{cdk}, which consists in introducing
a one-parameter family of L\'evy processes $(X^\lambda_t)_{t\geq 0}$,
where $\lambda = 1$ corresponds to the model of interest, and $\lambda
= 0$ corresponds to the Brownian motion. One can then expand the
quantities of interest in $\lambda$ around $\lambda = 0$. In \cite{cdk} such
expansions were developed for the expectation and the residual risk
of the quadratic hedging strategy in the case of European options with smooth $C^\infty$
pay-offs. In our paper we employ a different technique to prove these expansions for a class of
non-smooth pay-offs including for example the European put option. It
is important to note that in problems of this type, regularity of the
pay-off is an essential property which can strongly influence the
convergence rate (compare for example with \cite{gobet.temam.01}). 

As a result, we obtain an approximate formula for the
indifference price of a European option in an exponential L\'evy model
as a linear combination of the Black-Scholes price and correction
terms which depend on {variance,} skewness and kurtosis {of the underlying L\'evy process} as well as high-order
derivatives and a simple integral functional of the Black-Scholes
option price. Our method is based on an interpolation between a Lévy
process and a Brownian motion and works well when the L\'evy process in question is ``not too
 far'' from the Brownian motion. In a numerical study we compare our
 approximate formula to the exact value obtained by solving the
 integro-differential equation in the Merton jump-diffusion model and
 show that the precision of the approximate formula is {good} for
 realistic parameter values. 

Our approximate formula can be seen as an extension to the nonlinear
utility indifference price of the valuation
methodology based on the expansion around a proxy model, which was developed {by}
E. Gobet and collaborators in the diffusion setting (see e.g., \cite{benhamou2009smart}) and in
\cite{cdk} in the setting of L\'evy processes. It is also related to the ``expansion in the model
space'' technique for utility optimization recently discussed in
\cite{larsen2014expansion}. 

An important by-product of our study is a simple explicit approximate formula for
the spread between the buyer's and the seller's indifference price:
$$
p_{s} - p_b \approx \underbrace{\alpha}_\text{Risk aversion} \times \underbrace{\frac{1}{4}\left(m_4 - \frac{m_3^2}{\bar
    \sigma^2}\right)}_\text{Lévy model}\times  \underbrace{\mathbb E^{BS}\left[\int_0^T\left(S_t^2 \frac{\partial^2
    P_{BS}(t,S_t)}{\partial S^2}\right)^2 dt\right]}_\text{Jump risk
sensitivity of the option},
$$
where $\bar \sigma^2$ is the variance of the L\'evy process at time $1$,
$m_3$ and $m_4$ are the third and fourth moments of the L\'evy
measure and $P_{BS}$ and
$\mathbb E^{BS}$ denote {the} option price {and} {the} expectation computed in
the Black-Scholes
model with volatility $\bar \sigma$.  In
our asymptotic regime, therefore, the spread is decomposed into a product of
three factors: the risk-aversion which characterizes the economic
agent, a factor depending only on the properties of the L\'evy model
and a factor depending only on the variance $\bar \sigma^2$ of the
price process
and on the
properties of the contingent claim whose price is being computed. The
last factor therefore provides a model independent measure of the
sensitivity of a given product to jump risk in the limit of {small jumps occurring at high frequency}. 

The remainder of the paper is structured as follows. In Section
\ref{math.sec} we present the precise mathematical setup for
exponential utility maximization and indifference pricing in
exponential L\'evy models. In Section \ref{approx.sec} we {derive} a
non-asymptotic approximation for the utility indifference price, and
in Section \ref{bs.sec} this formula is used to develop an expansion
for the price ``in the neighborhood of the Black-Scholes
model''. Section \ref{num.sec} presents a numerical study of the
performance of our expansion, and Section \ref{bidask.sec} analyzes
specifically the formula for the bid-ask spread. Finally, proofs of
technical lemmas are relegated to the appendices.

\section{Mathematical framework}
\label{math.sec}


\subsection*{Exponential Lévy models}
Let $X$ be a L\'evy process on a probability space $(\Omega, \mathcal
F, \mathbb P)$, and let $(\mathcal F_t)_{t\geq 0}$ be the completed
natural filtration of $X$. We fix a time horizon $T<\infty$ and
consider a price process $S$ defined for $t\in [0,T]$
by $S_t =S_0\mathcal{E}(X)_t$ where $S_0>0$ is a constant and
$\mathcal E$ denotes the Doléans-Dade exponential defined by
$$
\mathcal E(X)_t = e^{X_t - \frac{1}{2}[X]^c_t} \prod_{0\leq s \leq t:
  \Delta X_s\neq 0} (1+\Delta X_s)e^{-\Delta X_s}. 
$$
The
interest rate is taken to be zero throughout the paper. We make the following standing assumption.

\begin{hyp}
\label{hyp1}
The process $X$ is not a.s.~monotone and there exists $\delta < 1$
such that $X$ satisfies $|\Delta X_t| \leq \delta $ a.s. for all $t \in [0,T]$. 
\end{hyp}

\begin{rem}
The lower bound on the jumps of $X$ ensures in particular that $S_t>0$
{a.s. for all $t \in [0,T]$.} The non-monotonicity ensures that there
exists a probability measure $\mathbb Q$ equivalent to $\mathbb P$
under which $S$ is a martingale, which guarantees absence of arbitrage in
the model (see e.g., \cite[section 9.5]{levybook}). The upper bound on the jumps is a technical assumption
needed in the subsequent developments. It is possible to assume that
$\Delta X_t \leq K$ for $K>1$, however for notational convenience we
impose the same bound on the negative and the positive jumps. 
\end{rem}


Denote by $(\sigma^2, \nu, \gamma)$ the
characteristic triplet of $X$ associated with the truncation function
$x \mapsto \mathbf 1_{|x|\leq 1}$. 
By the L\'evy-Khintchine formula, this means that the for all $t\in
[0,T]$, the characteristic function of $X_t$ is
given by
\begin{align*}
\mathbb E[e^{iu X_t}] = e^{t\psi(u)},\quad \psi(u) = iu\gamma - \frac{\sigma^2 u^2}{2} + \int_{\mathbb R}(e^{iux}-1-iux\mathbf 1_{|x|\leq 1}) \nu(dx).
\end{align*}

 
\subsection*{Utility indifference pricing}
Consider a bounded contingent claim $H = h(S)\in \mathcal F_T$, where
$h:\mathbb D([0,T])\to \mathbb R$ is a mapping defined on the
space of c\`adl\`ag trajectories. When the financial market is not complete, this claim cannot be perfectly
replicated, and therefore {the price at which an individual agent will
accept to buy / sell the claim will depend on the agent's attitude
towards risk, which may be quantified by a utility function. } 



In this paper, we focus on the exponential utility function defined by
$U(x) = - e^{-\alpha x}$ where $\alpha \in (0,\infty)$ is the risk
aversion parameter. Define the set of admissible trading strategies
\begin{equation*}
\Theta = \{\vartheta \in L(S) \ | \ \exists L^* \text{ with } \mathbb{E}[e^{-\alpha L^*}] <\infty 
\text{ s.t. } (\vartheta  \cdot  S)_t \geq L^*\, \forall t\in[0,T]\ \text{a.s.}\}
\end{equation*}
where $L(X)$ is the set of $\mathcal{F}$-predictable $X$-integrable $\mathbb{R}$-valued processes, 
and $(\vartheta  \cdot  X)_t := \int_0^t \vartheta_s dX_s$ denotes the
stochastic integral with respect to $X$. Other definitions of the set
of admissible strategies have been suggested in the literature
\cite{sixauthor}, but the above one appears sufficient in the context
of exponential L\'evy models and it is somewhat more elementary and
easier to check than the ones in \cite{sixauthor}. 

As mentioned in the introduction, the seller's utility indifference price of the
claim $H$ is defined by the implicit relation \eqref{priceimp}, which,
in the case of the exponential utility function, yields the explicit
formula
\begin{align}
p^H_s=\frac{1}{\alpha}\log \frac{\min_{\vartheta\in \Theta } \mathbb
E\left[\exp\left(  -\alpha(\vartheta \cdot S)_T
    {+ \alpha H_T}\right)\right]}{\min_{\vartheta\in \Theta }\mathbb E\left[\exp\left(-\alpha (\vartheta
    \cdot S)_T\right)\right] }.\label{iprice}
\end{align}
The buyer's indifference price $p^H_b$ is defined in a similar manner
and satisfies $p^H_b = -p^{-H}_s$. In the sequel, we shall focus on the
seller's price and omit the indices $H$ and $s$. 

In conclusion, to compute the utility indifference price, we \textit{a
  priori} need to solve two optimization problems. As shown
below, in the
context of exponential L\'evy models, the denominator of
\eqref{iprice} can be computed explicitly. 


Let 
\begin{equation*}
\ell(u) = \gamma u + \frac{\sigma^2 u^2}{2} + \int_{\mathbb R}(e^{ux}-1-ux)\nu(dx).
\end{equation*}
Under the assumption that $X$ is not a.s. monotone and has bounded
jumps, $\ell(u)$ is well defined for all $u\in \mathbb R$, is 
bounded from below and there exists $\varphi^*\in \mathbb R$ such that
$\ell(-\alpha \varphi^*) = \inf_u \ell(u)$ (see the proof of Theorem 1
in \cite{tankov.pplnmf}). 
Let $\vartheta \in \Theta$ and $\varphi_t = {\vartheta_t}{S_{t-}}$.
By It\={o} formula it is easy to show that 
\begin{equation*}
M_t = e^{-\alpha \int_0^t \varphi_s dX_s - \int_0^t \ell(-\alpha\varphi_s) ds }
\end{equation*}
is a local martingale. In addition, it is positive and bounded from
above by $e^{-\alpha L^*-T\inf_u \ell(u)}$, hence a true martingale. Therefore
\begin{align*}
\mathbb E[e^{-\alpha (\vartheta\cdot
S)_T}] 
&= \mathbb E[e^{-\alpha\int_0^T \varphi_t dX_t - \int_0^T \ell(-\alpha \varphi_t)dt 
+ \int_0^T \ell(-\alpha \varphi_t)dt}]\\  
&\geq \mathbb E[e^{-\alpha\int_0^T \varphi_t dX_t - \int_0^T \ell(-\alpha \varphi_t)dt }] e^{\int_0^T \ell(-\alpha \varphi^*)dt} 
= e^{\int_0^T \ell(-\alpha \varphi^*)dt}.
\end{align*}
On the other hand, the strategy $\vartheta_t^* = \frac{\varphi^*}{S_{t-}} $ is admissible. Indeed, as $X$ has bounded jumps, its exponential moments are finite, and by {Theorem} 25.18 from Sato \cite{Sato}, $\mathbb{E}[e^{-\alpha \varphi^* \inf_{0\leq t \leq T}X_t}] < \infty$. Thus, taking the strategy $\vartheta^*$ we get equality in the above inequality, which shows that this strategy is optimal. Then we define 
\begin{equation}
\frac{d\mathbb{Q}^*}{d\mathbb{P}} = \frac{e^{-\alpha \varphi^* X_T}}{\mathbb{E}[e^{-\alpha \varphi^* X_T}]}.\label{memm}
\end{equation}
Note that by definition of $\varphi^*$, the measure $\mathbb Q^*$ does not depend
on $\alpha$. 
Using {Theorem} 33.1 from \cite{Sato}, it can be shown that under
$\mathbb{Q^*}$, $X$ is a martingale Lévy process with diffusion
component volatility $\sigma$ and L\'evy measure $\nu^*$, where $\nu^*(dx) = e^{-\alpha \varphi^* x} \nu(dx)$.
In particular, it implies that $S$ is a $\mathbb{Q}^*$
martingale. The measure $\mathbb{Q}^*$ is the MEMM for $S$ (see \cite{FujiwaraMiyahara}). Using this  measure, the utility indifference price writes:
\begin{equation}
p= \frac{1}{\alpha} \log \inf_{\vartheta \in \Theta} \mathbb{E}^* e^{-\alpha ((\vartheta\cdot
S)_T- H)}.\label{indprice.eq}
\end{equation}

\subsection*{Quadratic hedging}
Quadratic ({also called} mean-variance) hedging consists in finding an initial capital and a hedging strategy which minimize the expected squared P\&L (Profit and Loss), that is:
\begin{equation*}
\min_{c\in \mathbb R,\vartheta \in \Theta'} \mathbb{E} [(c+(\vartheta \cdot S)_T - H)^2],
\end{equation*}
where $\Theta'$ is a suitable class of admissible strategies. 
We refer to \cite{pham}, \cite{guided}, \cite{general} for more
details on quadratic hedging and to \cite{kallsen.hubalek.al.06} for
the specific setting of exponential L\'evy models.

We shall
see that the exponential utility indifference price is
closely related to quadratic hedging under the measure $\mathbb
Q^*$. Since $\mathbb Q^*$ is a martingale measure, we can define
$$
\Theta' = \{\vartheta \in L(S): \vartheta \cdot S \  \text{is a square
  integrable martingale}\},
$$ 
and the optimal strategy can be computed as
$$
\bar \vartheta_t = \frac{d\langle S, H\rangle^{\mathbb Q^*}_t}{d\langle S,
  S\rangle^{\mathbb Q^*}_t},\quad \text{with}\quad H_t = \mathbb E^{*}[H|\mathcal F_t]. 
$$

However, it should be noted that the optimal initial capital $c$
may not be interpreted as a price of the claim $H$ since it is equal
to the price of the hedging strategy only and does not take into
account the unhedged residual risk. 


\section{An approximation for the indifference price}
\label{approx.sec}

The goal of this section is to obtain a non-asymptotic approximation
for the exponential utility indifference price in terms of the
quadratic {residual risk (error)} under $\mathbb{Q^*}$ in the Lévy
model under consideration. The quadratic hedging strategy under
$\mathbb Q^*$ will be denoted by $\bar\vartheta$. 

Since $X$ is a L\'evy process, the $\mathbb Q^*$-martingale
$(H_t)_{0\leq t \leq T}$ has the predictable representation property
\cite[paragraph III.4d]{jacodshiryaev} and can be written as 
\begin{align}
H_t = \mathbb E^{*}[H] +\int_0^t \sigma_s dX^c_s + \int_0^t \int_{\mathbb R}
\gamma_s(z) \tilde J_X(ds\times dz),\label{martrep.eq}
\end{align}
where $X^c$ is the continuous martingale part of the process $X$ under
$\mathbb Q^*$, $\tilde J_X$ is the compensated jump measure of the
process $X$ under $\mathbb Q^*$, $\sigma_t$ is a predictable process
and $\gamma_s(\cdot)$ is a predictable random function. 

\begin{thm}
\label{thm1}
Assume that
there exists a constant $L $ with $2\delta L \alpha < 1$ such that 
\begin{align}
&|H-\mathbb E^*[H]|\leq L\quad \text{a.s.}, \label{ass1}\\
&|\sigma_t|\leq L\quad \text{a.s. for all $t\in[0,T]$},\label{ass2}\\
&|\gamma_t(z)| \leq L|z|\quad \text{a.s. for all $t\in[0,T]$ and all
  $z\in \supp \nu$.}\label{ass3}
\end{align}
Then there exists a constant 
$C_{\alpha\delta L}<\infty$ such that for every $\varepsilon\in(0,1]$ the seller's indifference
price of the claim $H$ satisfies
\begin{equation*}
\left| p - \mathbb E^*[H] - \frac{\alpha}{2} \mathbb
  E^*\left[\left(\int_0^T \bar\vartheta_s dS_s - (H-\mathbb E^*[H])\right)^2\right] \right|
\leq \alpha^{1+\varepsilon} C_{\alpha\delta L} \mathbb E^*\left[\sup_{0\leq t\leq T}\left|\int_0^t \bar\vartheta_s
    dS_s -(H_t-\mathbb E^*[H])\right|^{2+\varepsilon}\right].
\end{equation*}
The constant
$C_{\alpha\delta L}$ can be chosen as
$$
C_{\alpha\delta L} = C e^{4\alpha L} \vee (1 - 2\alpha \delta L)^{-2},
$$
where $C<\infty$
is a universal constant. 
\end{thm}
\begin{rem}
The formula of the above theorem is a non-asymptotic approximation
formula for the indifference price, which can be used to recover a
variety of asymptotic results. For example, observing that
$C_{\alpha\delta L} $ is bounded as $\alpha\to0$, we
recover the asymptotics for small risk aversion:
$$
p = \mathbb E^*[H] + \frac{\alpha}{2} \mathbb
  E^*\left[\left(\int_0^T \bar\vartheta_s dS_s - (H-\mathbb
      E^*[H])\right)^2\right] + o(\alpha),\quad \alpha \to 0.
$$
\end{rem}


\begin{proof}
Under the assumptions of the Theorem,
$$
\bar \vartheta_t = \frac{1}{S_{t-}}\frac{\sigma \sigma_t +\int_{\mathbb R}
  z\gamma_t(z)\nu(dz)}{ \sigma^2 + \int_{\mathbb R}z^2\nu(dz) }
$$
and therefore $|S_{t-} \bar\vartheta_t|\leq L$ a.s. for all $t\in [0,T]$. 
 We assume without loss of generality that $\mathbb
  E^*[H] =0$. 

Applying first Lemma \ref{ub.lm} below with pay-off $H' =
  \alpha H$ and bound $L' = \alpha L$, taking the logarithm using the
  inequality $\log(1+x)\leq x$ and dividing by $\alpha$, we get
$$
p\leq \frac{\alpha}{2} \mathbb E^*\left[\left(\int_0^T \bar\vartheta_s
    dS_s -H\right)^2\right] + {
C \alpha^{1+\varepsilon}e^{4\alpha L}}
\mathbb E^*\left[\sup_{0\leq t\leq T}\left|\int_0^t \bar\vartheta_s dS_s -H_t\right|^{2+\varepsilon}\right]{.}
$$
Similarly, applying Lemma \ref{lb.lm} below yields
$$
p\geq \frac{\alpha}{2} \mathbb E^*\left[\left(\int_0^T \bar\vartheta_s
    dS_s -H\right)^2\right] 
    - {
C \alpha^{1+\varepsilon}
(1 - 2\alpha \delta L)^{-2}}
\mathbb E^*\left[\sup_{0\leq t\leq T}\left|\int_0^t \bar\vartheta_s dS_s -H_t\right|^{2+\varepsilon}\right].
$$
\end{proof}

\begin{lem}[Upper Bound]\label{ub.lm}
Let $L$ be a constant such
that assumptions \eqref{ass1}--\eqref{ass3} are satisfied.  Then,
there exists a universal constant $C$ such that for every $\varepsilon\in(0,1]$,
\begin{align*}
\inf_{\vartheta\in \Theta}\mathbb E^*\left[e^{-\int_0^T \vartheta_s dS_s +H}\right] 
\leq 1+\frac{1}{2} \mathbb E^*\left[\left(\int_0^T \bar\vartheta_s dS_s -H\right)^2\right] + C e^{4L}\mathbb E^*\left[\sup_{0\leq t\leq T}\left|\int_0^t \bar\vartheta_s dS_s -H_t\right|^{2+\varepsilon}\right].
\end{align*}
\end{lem}

\begin{proof} 
Introduce the stopping time
\begin{equation*}
\tau = \inf\{t\geq 0: \left|\int_0^t \bar\vartheta_s dS_s -
  H_t\right|\geq 1\}\wedge T.
\end{equation*}
Since by assumptions, taking into account that $\delta<1$, 
\begin{equation*}
\left|\int_0^\tau \bar\vartheta_s dS_s - H\right| \leq  3L+1,
\end{equation*}
the strategy $\vartheta_t = \bar\vartheta_t  1_{t\leq \tau}$ belongs
to $\Theta$ and we get
\begin{equation*}
\inf_{\vartheta\in\Theta}\mathbb E^*\left[e^{-\int_0^T \vartheta_s dS_s + H}\right]
\leq  \mathbb E^*\left[e^{-\int_0^\tau \bar\vartheta_s dS_s + H}\right].
\end{equation*}
We shall use a Taylor formula of the following form: for every
$\varepsilon\in (0,1)$ and $m <\infty$, 
\begin{equation*}
e^{x} \leq 1+x + \frac{x^2}{2} + C |x|^{2+\varepsilon},\quad \forall x \in[-m,m]
\end{equation*}
with 
$$
C = \frac{m^{1-\varepsilon}e^{m}}{6} .
$$
Thus, for $C_{L\varepsilon} =\frac{(3L+1)^{1-\varepsilon}e^{3L+1}}{6} $ 
\begin{align*}
\inf_{\vartheta\in \Theta}\mathbb E^*\left[e^{-\int_0^T {\vartheta_s} dS_s +H\}}\right] 
&\leq 1 + \frac{1}{2} \mathbb E^*\left[\left(\int_0^\tau \bar\vartheta_s dS_s -H\right)^2\right] 
+ C_{L\varepsilon} \mathbb E^*\left[\left|\int_0^\tau \bar\vartheta_s dS_s - H\right|^{2+\varepsilon}\right]\\
& \leq 1 + \frac{1}{2} \mathbb E^*\left[\left(\int_0^T \bar\vartheta_s dS_s -H\right)^2\right] 
+ C_{L\varepsilon}\mathbb E^*\left[\left|\int_0^T \bar\vartheta_s dS_s -H\right|^{2+\varepsilon}\right] \\ 
&\qquad+ (3L+1)^2 \left( \frac{1}{2} +C_{L\varepsilon} (3L+1)^{\varepsilon}\right)  { \mathbb Q^*}[\tau < T]. 
\end{align*}
Then, by Markov inequality
\begin{equation*}
 { \mathbb Q^*}[\tau < T] =  { \mathbb Q^*}\left[\sup_{0\leq t \leq T}\left|\int_0^t \bar\vartheta_s dS_s - H_t\right|>1\right] 
\leq \mathbb E^*\left[\sup_{0\leq t \leq T}\left|\int_0^t \bar\vartheta_s dS_s - H_t\right|^{2+\varepsilon}\right]
\end{equation*}
so that 
\begin{multline*}
\inf_{\vartheta\in\Theta}\mathbb E^*\left[e^{-\int_0^T {\vartheta_s} dS_s
    +H\}}\right] \leq 1 + \frac{1}{2} \mathbb E^*\left[\left(\int_0^T
    \bar\vartheta_s dS_s -H\right)^2\right] \\
 + (C_{L\varepsilon} + \frac{(3L+1)^2}{2} + C_{L\varepsilon}  (3L+1)^{2+\varepsilon}))\mathbb E^*\left[\sup_{0\leq t \leq T}\left|\int_0^t \bar\vartheta_s dS_s - H_t\right|^{2+\varepsilon}\right].
\end{multline*}
Now, it is clear that one can choose a universal constant $C$ such
that the statement of the Lemma holds true. 
\end{proof}

\begin{lem}[Lower Bound]\label{lb.lm}
Let $L$ be a constant with $2L\delta<1$ such
that assumptions \eqref{ass1}--\eqref{ass3} are satisfied.
Then, there exists a universal constant $C$ such that for every $\varepsilon\in(0,1]$,
\begin{equation*}
p \geq \frac{1}{2}  \mathbb E^*\left[\left(\int_0^T \bar\vartheta_s dS_s - H\right)^2\right] 
- \frac{C}{(1-2L\delta)^2} \mathbb E^*\left[\sup_{0\leq t\leq T}\left|\int_0^t \bar\vartheta_s dS_s -H_t\right|^{2+\varepsilon}\right]
{.}
\end{equation*}
\end{lem}

\begin{proof}
From the results of \cite{bellini.frittelli.02}, we have\footnote{This
reference provides a duality result for the class of admissible
strategies which are bounded from below, but it can easily be
extended to our class $\Theta$ using the dominated convergence theorem
and the local boundedness of $S$.}
\begin{equation*}
p = \sup_{\mathbb Q\in \text{EMM}(\mathbb Q^*)}\left\{\mathbb
E^{\mathbb Q}[H] - H(\mathbb Q| \mathbb Q^*)\right\},
\end{equation*}
where EMM $(\mathbb Q^*)$ denotes the set of martingale measures,
equivalent to $\mathbb Q^*$ and $H(\mathbb Q| \mathbb Q^*)$ is defined by
\begin{equation*}
H(\mathbb Q| \mathbb Q^*) 
= \mathbb E^*\left[\frac{d\mathbb Q}{d\mathbb Q^*} \log \frac{d\mathbb Q}{d\mathbb Q^*}\right]
\end{equation*}
whenever this quantity is finite and equals $+\infty$ otherwise. 
Therefore, for any random variable $D>0$ such that $D \mathbb Q^*$ is a martingale measure,
\begin{align}
p\geq \mathbb E^{*}[D H_T] - \mathbb E^*[D \log D]. \label{rhs}
\end{align}
Let $\kappa = \frac{1}{2} - L\delta$,
introduce the stopping time 
\begin{equation*}
\tau_\kappa = \inf\{t\geq 0: \left|\int_0^t \bar\vartheta_s dS_s - H_t\right|\geq \kappa \}\wedge T
\end{equation*}
and define
\begin{equation*}
D = 1 + \int_0^{\tau_\kappa} \bar\vartheta_t dS_t - H_{\tau_\kappa}. 
\end{equation*}
By construction, $\mathbb E[D] =1$ and 
$$
|D-1|\leq \kappa + |\bar \vartheta_{\tau_\kappa-}
S_{\tau_{\kappa}-}\Delta X_{\tau_{\kappa}}|+  |\Delta
H_{\tau_\kappa}|\leq \kappa + 2L \delta \leq \frac{1}{2} + L\delta < 1.
$$
Moreover, for a bounded strategy $\vartheta$,
\begin{align*}
\mathbb E^*\left[D \int_0^T \vartheta_t dS_t\right] 
= \mathbb E^*\left[D \int_0^{\tau_\kappa} \vartheta_t dS_t\right] 
= \mathbb E^*\left[\left(\int_0^T \bar\vartheta_t dS_t - H_T\right) \int_0^{\tau_\kappa} \vartheta_t dS_t\right]
= 0
\end{align*}
because $\bar\vartheta$ is the optimal quadratic {hedging} strategy. Therefore, $D\mathbb Q^*$ is a martingale measure. 
It remains to compute the right-hand side of \eqref{rhs}. For the first term, using the Cauchy-Schwarz inequality and an estimate 
for $\tau_\kappa$, we get

\begin{align*}
&\mathbb E^{*}[D H_T] = \mathbb E^{*}[D H_{\tau_\kappa}] 
= \mathbb E^*\left[\left(H_{\tau_\kappa} - \int_0^{\tau_\kappa} \bar\vartheta_t dS_t\right)^2\right]\\
&\geq  \mathbb E^*\left[\left(H_T - \int_0^T \bar\vartheta_t dS_t\right)^2\right] 
- \mathbb E^*\left[\left(H_T - \int_0^T \bar\vartheta_t dS_t\right)^2 \mathbf 1_{{\tau_\kappa} < T}\right] \\
&\geq  \mathbb E^*\left[\left(H_T - \int_0^T \bar\vartheta_t dS_t\right)^2\right] 
-  \mathbb E^*\left[\left|H_T - \int_0^T \bar\vartheta_t dS_t\right|^{2+\varepsilon} \right]^{\frac{2}{2+\varepsilon}} 
\mathbb P[{\tau_\kappa} < T]^{\frac{\varepsilon}{2+\varepsilon}}\\
&\geq  \mathbb E^* \left[\left(H_T - \int_0^T \bar\vartheta_t dS_t\right)^2\right] - 
\frac{1}{\kappa^\varepsilon}\mathbb E^*\left[\sup_{0\leq t \leq T}\left|\int_0^t \bar\vartheta_s dS_s - H_t\right|^{2+\varepsilon}\right].
\end{align*}
The second term in \eqref{rhs} can be estimated using the following
Taylor formula: for every $\varepsilon \in (0,1)$ and $\Delta \in
(0,1)$, 
\begin{equation*}
x\log x + 1- x \leq \frac{(x-1)^2}{2} + C |x-1|^{2+\varepsilon},\quad
\forall x \in [1-\Delta,1+\Delta]\quad \text{with}\quad C = \frac{\Delta^{1-\varepsilon}}{6(1-\Delta)^2}.
\end{equation*}
Then, for $ { C_{L\delta\varepsilon} }= \frac{(1/2+L\delta)^{1-\varepsilon}}{6(1/2-L\delta)^2}$,
\begin{align*}
&\mathbb E^*[D\log D] = \mathbb E^*[D\log D + 1-D] \leq \frac{1}{2} \mathbb E^*[(D-1)^2] 
+ C_{L\delta\varepsilon} \mathbb E^*[|D-1|^{2+\varepsilon}]\\
& = \frac{1}{2} \mathbb E^*\left[\left(H_{\tau_\kappa} - \int_0^{\tau_\kappa} \bar\vartheta_t dS_t\right)^2\right] 
+ C_{L\delta\varepsilon}  \mathbb E^* \left[\left|H_{\tau_\kappa} - \int_0^{\tau_\kappa} \bar\vartheta_t dS_t \right| ^{2+\varepsilon}\right] \\
&\leq \frac{1}{2} \mathbb E^*\left[\left(H_T - \int_0^T \bar\vartheta_t dS_t\right)^2\right] 
+ (C_{L\delta\varepsilon}+\frac{1}{\kappa^\varepsilon})\mathbb E^*\left[\sup_{0\leq t \leq T}\left|\int_0^t \bar\vartheta_s dS_s -H_t\right|^{2+\varepsilon}\right].
\end{align*}
Adding up the estimates for the first and the second term of
\eqref{rhs}, and choosing the universal constant $C$ appropriately, the proof of the Lemma is complete. 
\end{proof}

\begin{ex}\label{ex1}
Let us check the assumptions of Theorem \ref{thm1} for the European put
option with pay-off $H = (K-S_T)^{{+}}$.
The process $(H_t)_{0\leq t\leq T}$ is given by{:}
\begin{equation*}
H_t := \mathbb{E}^* [H | \mathcal{F}_t] = P(t,S_t) 
\end{equation*}
where $P(t,S) = \mathbb{E}^* [(K-S\mathcal{E}(X)_{T-t})^{{+}}]$, and under
suitable regularity assumptions on the process $X$ (see e.g.,
\cite[Proposition 2]{options}), we have the martingale representation
$$
H_t = \mathbb E^*[H] + \int_0^t \sigma_t dX^c_t + \int_0^t \int_{\mathbb
  R}\gamma_s(z)\tilde J(ds\times dz)
$$
with
$$
\sigma_t = \frac{\partial P(t,S_t)}{\partial S}S_{t}\quad
\text{and}\quad \gamma_t(z) = P(t,S_{t-}(1+z)) - P(t,S_{t-}).
$$
By dominated convergence{:}
\begin{equation*}
\left| S\frac{\partial P(t,S)}{\partial S} \right|
= \mathbb{E}^* [S\mathcal{E}(X)_{T-t} 1_{S\mathcal{E}(X)_{T-t} \leq K}] \leq K.
\end{equation*}
On the other hand, for $z\in \supp \nu$, 
\begin{equation*}
|P(t,S(1+z))-P(t,S)| \leq \mathbb
E^*[|zS\mathcal{E}(X)_{T-t}|1_{S\mathcal{E}(X)_{T-t}(1+z)\wedge 1 \leq
  K}]\leq \frac{K|z|}{1-\delta}.   
\end{equation*}

\end{ex}
\section{Indifference price asymptotics in the neighborhood of the
  Black-Scholes model}
\label{bs.sec}
Since, as we have seen, the computation of the indifference price can
be carried out under the MEMM, in this section, to simplify notation
we omit the star in $\mathbb E^*$. In other words, we simply assume
that all the expectations are taken under the MEMM unless specified
otherwise, and that $X$ is a martingale L\'evy process with diffusion
component volatility $\sigma$ and L\'evy measure $\nu$. 

In liquid financial markets, jumps are typically small and in most
cases the
Black-Scholes model provides a correct ``order of magnitude''
approximation to option prices. Thus it seems reasonable, in these
markets, to treat more complex stochastic models as perturbations of
the Black-Scholes price and to compute correction terms to this
reference value. 
Our goal in this section is therefore to find an explicit approximation to the
indifference price \eqref{indprice.eq} in the situation when the
L\'evy process $X$ is ``close'' to the Brownian motion. 

To quantify what it means to be close to the Brownian motion, and following a recent paper by \v{C}ern\'{y}, Denkl and Kallsen
\cite{cdk}, we artificially introduce a small parameter
$\lambda \in (0,1)$ into the model, by considering the family of
stochastic processes 
\begin{equation*}
X^{\lambda}_t := \lambda X_{t/\lambda^2},\quad 0\leq t\leq T.
\end{equation*}
Note that our parameterization is slightly different from the one introduced in
\cite{cdk} because that paper considers L\'evy models built using
ordinary exponential, whereas we use the Dol\'eans-Dade
exponential. As a result, our formulas are somewhat simpler than the
ones of \cite{cdk}.

With this parameterization, 
$X^{1} = X$ and, as $\lambda\downarrow 0$,
$X^{\lambda}$ converges weakly in Skorokhod topology to the process
\begin{equation*}
(\overline X_{{t}})_{t\geq 0} = (\bar\sigma
W_{{t}})_{t\geq 0},
\end{equation*}
where $W$ is a standard Brownian motion and $\bar\sigma^2 = \sigma^2 +\int_{\mathbb R}x^2 \nu(dx)$. 
We then define:
\begin{equation*}
S^\lambda = S_0\mathcal{E}(X^\lambda).
\end{equation*}
{Similarly, it is easy to show that the process $S^\lambda$ convergences to $S_0\mathcal{E}( \bar\sigma W)$ as $\lambda$ tends to $0$.}

Let $H^\lambda=h(S^\lambda)$ and consider the corresponding
indifference price
\begin{equation}
p_\lambda= \frac{1}{\alpha} \log \inf_{\vartheta \in \Theta} \mathbb{E}\, e^{-\alpha ((\vartheta\cdot
S^\lambda)_T- H^\lambda)}\label{lambdaprice.eq}
\end{equation}
The following theorem provides an approximation of $p_\lambda$ when
$\lambda\to 0$ for European pay-offs, that is, we assume that $H =
h(S_T)$ and $H^\lambda = h(S^\lambda_T)$. In this theorem and below,
we let $P_{BS}(t,S)$ denote the Black-Scholes price of the
corresponding option computed with volatility $\bar\sigma$ defined by $\bar \sigma^2 =
\sigma^2 + \int x^2 \nu(dx)$, that is,
$$
 { P_{BS}(t,S)} = \mathbb E\left[h\left(Se^{-\frac{\bar\sigma^2}{2}(T-t) +
    \bar \sigma W_{T-t}}\right)\right],
$$
where $W$ is a standard Brownian motion. When $t=0$ we also write 
{$P_{BS}(0,S) = P_{BS}(S)$} to shorten notation.

\begin{thm}\label{thm2}
Assume that 
\begin{itemize}
\item The pay-off function $h$ is a bounded, almost everywhere
  differentiable, the derivative $h'$ has finite variation on
  $[0,\infty)$ and there exists $L<\infty$ such that $|xh'(x)|\leq L$
  almost everywhere.
\item Either $\sigma >0$ or there exists $\beta \in (0,2)$ such that $\liminf\limits_{r\downarrow 0} \frac{\int_{[-r,r]} x^2 \nu(dx)}{r^{2-\beta}} > 0$. 
\end{itemize}
Then, as $\lambda \to 0$, 
\begin{align*}
p_\lambda = &P_{BS}(S_0) 
+ \frac{\lambda m_3 T}{6} S_0^3 P_{BS}^{(3)}(S_0) + \frac{\lambda^2 m_4T}{24} S_0^4P_{BS}^{(4)}(S_0)\\
&+ \frac{\lambda^2m_3^2T^2}{72}
\left\lbrace  6S_0^3P_{BS}^{(3)}(S_0)+18S_0^4P_{BS}^{(4)}(S_0)
+9S_0^5P_{BS}^{(5)}(S_0)+S_0^6P_{BS}^{(6)}(S_0) \right\rbrace\\& + \frac{\alpha\lambda^2}{8}\left(m_4-\frac{m_3^2}{\bar{\sigma}^2} \right) \mathbb{E}^{BS} \left[ \int_0^T \left(S_t^2 \frac{\partial^2P_{BS}(t,S_t)}{\partial S^2} \right)^2 dt\right]+o(\lambda^2)
\end{align*}
where $m_3 = \int_\mathbb{R}x^3 \nu(dx)$, $m_4 = \int_\mathbb{R}x^4
\nu(dx)$ and $\mathbb E^{BS}$ denotes the expectation computed in the
Black-Scholes model with volatility $\bar\sigma$.
\end{thm}
\begin{rem}
It is easy to check that the pay-off function of the European put
option $h(x) = (K-x)^+$ satisfies the first assumption of the
Theorem. Moreover, this assumption implies that assumptions of Theorem
\ref{thm1} (by an argument similar to the one given in Example \ref{ex1}). As for the second assumption, it is satisfied by most
parametric L\'evy models used in practice, such as CGMY
\cite{finestructure} (with $Y>0$) and normal
inverse Gaussian \cite{Barndorff-Nielsen}. It is not satisfied by the
variance gamma model \cite{madan98}. 
\end{rem}
\begin{proof}
The proof is based on the non-asymptotic approximation formula of Theorem
\ref{thm1}, applied to the process $S^\lambda$, which takes the
form
\begin{equation*}
\left| p_\lambda - \mathbb E[H^\lambda] - \frac{\alpha}{2} \mathbb
  E\left[\left(\int_0^T \bar\vartheta_s^\lambda dS^\lambda_s -
      (H^\lambda-\mathbb E[H^\lambda])\right)^2\right] \right|
\leq  C \mathbb E\left[\sup_{0\leq t\leq T}\left|\int_0^t
    \bar\vartheta_s^\lambda dS_s -(H^\lambda_t-\mathbb E[H^\lambda])\right|^{2+\varepsilon}\right],
\end{equation*}
where $\bar{\vartheta}^\lambda = \argmin_\vartheta  { \mathbb{E}}
\left[\left(\int_0^T \vartheta_t dS_t^\lambda
    -(H^\lambda-\mathbb E[H^\lambda])\right)^2\right]$. The following lemmas provide
estimates of the linear part of the price $\mathbb E[H^\lambda]$,
the nonlinear part of the price $\mathbb E\left[\left(\int_0^T
    \bar\vartheta_s^\lambda dS^\lambda_s - (H^\lambda-\mathbb E[H^\lambda])\right)^2\right]$
and the residual term in the right-hand side. In these lemmas we
suppose that the standing assumptions of the paper hold true. Note
that the expansion of the linear part of the price does not require
the pay-off function $h$ to be regular, but in the other two lemmas,
regularity is an essential assumption without which the convergence
rates as $\lambda \downarrow 0$ may be different. 

\begin{lem}[Estimation of the residual term]\label{residual.lm}
Let the assumptions of Theorem \ref{thm2} hold true, 
let $M_t^\lambda = \int_0^t \bar{\vartheta}^{\lambda}_s dS_s^\lambda -
(H_t^\lambda-\mathbb E[h(S^\lambda_T)])$ and define $\bar{M}_T^{\lambda} = \sup_{0\leq t\leq T}
|M_t^\lambda|$. Then $\forall q>2$, as $\lambda \to 0$
\begin{equation*}
\mathbb E\left[(\bar{M}_T^\lambda)^{q}\right] = O\left(\lambda^q \left(\log \frac{1}{\lambda}\right)^{\frac{q}{2}}\right). 
\end{equation*}
\end{lem}

\begin{lem}[Estimation of the nonlinear part of the
  price]\label{nonlinear.lm}
Let the assumptions of Theorem \ref{thm2} hold true.
Then, as $\lambda \to 0${,}
\begin{equation*}
\mathbb{E} \left[\left(\int_0^T \bar\vartheta^\lambda_t dS^\lambda_t -h(S^\lambda_T)\right)^2\right]
= \frac{\lambda^2}{4}\left(m_4-\frac{m_3^2}{\bar{\sigma}^2} \right)
\mathbb{E}^{BS} \left[ \int_0^T \left(S_t^2 \frac{\partial^2P_{BS}(t,S_t)}{\partial S^2} \right)^2 dt\right] +o(\lambda^2).
\end{equation*}

In addition, for the European put option with pay-off function $h(S_T)
= (K-S_T)^+$, 
{
\begin{equation*}
\mathbb{E}^{BS} \left[ \int_0^T \left(S_t^2 \frac{\partial^2P_{BS}(t,S_t)}{\partial S^2} \right)^2 dt\right]
= \frac{K^2}{2\pi\bar{\sigma}^2}  \int_0^1 \frac{e^{\frac{-d^2}{1+u}}du}{\sqrt{1-u^2}}
\end{equation*}
where $d = \frac{\log{\frac{S_0}{K}}-\frac{\bar{\sigma}^2}{2}T}{\bar{\sigma}\sqrt{T}}$.}

\end{lem}

\begin{lem}[Estimation of the linear part of the
  price]\label{linear.lm}
Assume that
\begin{itemize}
\item The function $h$ is measurable with polynomial growth.
\item Either $\sigma >0$  or there exists $\beta \in (0,2)$ such that $\liminf\limits_{r\downarrow 0} \frac{\int_{[-r,r]} x^2 \nu(dx)}{r^{2-\beta}} > 0$.
\end{itemize}
Then, as $\lambda \to 0${,}
\begin{align*}
\mathbb{E}[h(S_T^\lambda)] &= P_{BS}(S_0) 
+ \frac{\lambda m_3 T}{6} S_0^3 P_{BS}^{(3)}(S_0) + \frac{\lambda^2 m_4T}{24} S_0^4P_{BS}^{(4)}(S_0)\\
&+ \frac{\lambda^2m_3^2T^2}{72}
\left\lbrace  6S_0^3P_{BS}^{(3)}(S_0)+18S_0^4P_{BS}^{(4)}(S_0)
+9S_0^5P_{BS}^{(5)}(S_0)+S_0^6P_{BS}^{(6)}(S_0) \right\rbrace+o(\lambda^2).
\end{align*}
\end{lem}
\end{proof}

\section{Numerical results}\label{num.sec}
In this section, we illustrate numerically the performance of the
asymptotic formula of Theorem \ref{thm2}, assuming that the asset
price is described by Merton's
jump-diffusion model \cite{merton} under $\mathbb Q^*$. Strictly speaking, this model does
not satisfy the standing assumptions of the paper because the
(log-normal) jumps are not bounded from above. However, in the
numerical implementation discussed below, the L\'evy measure is
truncated to a bounded domain (which can be chosen sufficiently large so
that further increase of the domain does not modify the price). 
\paragraph{Merton's jump-diffusion model}

In this model the stock price is defined by $ S_t = S_0 \mathcal{E}(X)_t$ where
\begin{equation*}
X_t =  \mu t + \sigma W_t + \sum_{i=1}^{N_t} (e^{Y_i}-1)
\end{equation*}

where $W$ denotes standard Brownian motion, jump sizes  $(Y_i) \sim
\mathcal{N}(\gamma,\delta^2)$ are i.i.d. random variables and
$(N_t)_{t\geq 0}$ is an independent Poisson process with intensity
$\lambda^M$ accounting for the number of jumps up to time $t$. 
The L\'evy measure of $X$ therefore has a density given by
$$
\nu(x) = \frac{\lambda^M \mathbf 1_{x>-1}}{\delta (x+1) \sqrt{2\pi}} e^{-\frac{(\log(x+1)-\gamma)^2}{2\delta^2}}.
$$



\paragraph{Implementation of the asymptotic formula}
In the numerical examples, we let $\lambda=1$ and approximate the
indifference price by 
\begin{align}
p = &P_{BS}(S_0) 
+ \frac{m_3 T}{6} S_0^3 P_{BS}^{(3)}(S_0) + \frac{ m_4T}{24} S_0^4P_{BS}^{(4)}(S_0)\notag\\
&+ \frac{m_3^2T^2}{72}
\left\lbrace  6S_0^3P_{BS}^{(3)}(S_0)+18S_0^4P_{BS}^{(4)}(S_0)
+9S_0^5P_{BS}^{(5)}(S_0)+S_0^6P_{BS}^{(6)}(S_0) \right\rbrace\notag\\& + \frac{\alpha\lambda^2}{8}\left(m_4-\frac{m_3^2}{\bar{\sigma}^2} \right) \mathbb{E}^{BS} \left[ \int_0^T \left(S_t^2 \frac{\partial^2P_{BS}(t,S_t)}{\partial S^2} \right)^2 dt\right]\label{approxnum}
\end{align}
Using the formula which has been justified asymptotically as
$\lambda\to 0$ for a finite nonzero value of $\lambda$ amounts to use
a second-order Taylor expansion of a function at zero
to approximate the value of this function at a point $x\neq 0$. The
quality of the approximation does not depend on the specific value of
$x$, but rather on the smoothness of the function between $0$ and
$x$. The numerical examples of this section show that the indifference
price is indeed smooth as function of $\lambda$ and that using the
formula of Theorem \ref{thm2} with $\lambda=1$ leads to a very precise
approximation. 

To evaluate the approximate indifference price, one needs to perform
three computations.
\begin{itemize}
\item Evaluate $\bar \sigma^2$ and the moments of the L\'evy measure
  $m_3$ and $m_4$. In Merton's model these quantities are easily
  computed from the explicit form of the L\'evy measure and are given
  by
\begin{align*}
\bar \sigma^2 &= \sigma^2 + \lambda^M\{e^{2\gamma + 2\delta^2} -
2e^{\gamma + \frac{\delta^2}{2}} +1\}\\
m_3 & = \lambda^M\{e^{3\gamma + \frac{9}{2}\delta^2} -
3 e^{2\gamma + 2\delta^2}  + 3 e^{\gamma + \frac{\delta^2}{2}} -1\}\\
m_4 &= \lambda^M\{e^{4\gamma + 8\delta^2} - 4 e^{3\gamma + \frac{9}{2}\delta^2} +
6 e^{2\gamma + 2\delta^2}  - 4 e^{\gamma + \frac{\delta^2}{2}} +1\}
\end{align*}
Remark that although the original model has four parameters (since
$\mu$ is fixed by the martingale condition), the asymptotic formula
only depends on three `group' parameters $\bar \sigma^2$, $m_3$ and
$m_4$. 
\item Evaluate the integral in the last line of \eqref{approxnum}. In our example we
  consider the put option and evaluate the more explicit form of the
  integral given in Lemma \ref{nonlinear.lm} using an
  numerical integration algorithm.
\item Evaluate the derivatives of the Black-Scholes option price with
  respect to the underlying up to
  order $6$. The exact explicit formulas for these derivatives are
  given in Appendix \ref{greeks.app}. 
\end{itemize}

\paragraph*{Partial integro-differential equation and the finite
  difference scheme}
In this paragraph we briefly describe the HJB equation for the
indifference price (see e.g., \cite{jaimungal2005pricing,wu2009pricing}) as well as the numerical scheme used to
solve it. This scheme is inspired by well-studied schemes for linear
integro-differential equations \cite{cont2005finite} and is provided here only for the purpose of illustrating
the asymptotic method; its full derivation and the study of its
accuracy is out of scope of the present paper. 

Let $H_T = (K-S_T)^+$ and assume that $S$ has the dynamics
\[\frac{dS_t}{S_{t-}} = dX_t,
\]
where $X$ is a martingale Lévy process with Lévy measure $\nu$ and
diffusion coefficient $\sigma$. Then, the indifference price $p(t,S)$
satisfies (omitting the arguments where possible to save space)
\begin{align*}
&0=\frac{\partial p}{\partial t} + \frac{S^2 \sigma^2}{2}
\frac{\partial^2 p}{\partial S^2}
+\int_{\mathbb R}\left(p(t,S(1+z))-p-Sz\frac{\partial p}{\partial S}\right)\nu(dz)\\&
+\min_\vartheta \Bigg\{\frac{\alpha S^2 \sigma^2}{2}\left(\vartheta - \frac{\partial p}{\partial S}\right)^2 
+\frac{1}{\alpha}\int_{\mathbb R}
\left( e^{\alpha(p(t,S(1+z))-p-Sz\vartheta) }-1-\alpha(p(t,S(1+z))-p-Sz \vartheta)\right)\nu(dz)\Bigg\} 
\end{align*}
with terminal condition $p(T,S) = (K-S)^+$. 
In log-variable $x = \log S$, introducing $P(t,x) = p(t,S)$, 
\begin{align*}
&0=\frac{\partial P}{\partial t} + \frac{\sigma^2}{2}
\left(\frac{\partial^2 P}{\partial x^2} - \frac{\partial P}{\partial x}\right) 
+\int_{\mathbb R}\left(P(t,x+z) - P - (e^z-1) \frac{\partial P}{\partial x}\right)\bar\nu(dz)\\&
+\min_\vartheta \Bigg\{\frac{\alpha \sigma^2}{2}\left(\vartheta - \frac{\partial P}{\partial x}\right)^2 
+\frac{1}{\alpha}\int_{\mathbb R} 
\left( e^{\alpha(P(t,x+z)-P-(e^z-1) \vartheta) }-1-\alpha(P(t,x+z)-P-(e^z-1) \vartheta)\right)\bar\nu(dz)\Bigg\}, 
\end{align*}
where $\bar \nu$ is the logarithmic transformation of $\nu$. 
\medskip

To discretize this equation we introduce a time grid $t_i = ih$,
$i=0,\dots,N$ with $h = \frac{T}{N}$, a space grid $x_j = x_0 + j d$, $j=0,\dots,2M$, and represent the Lévy measure $\bar \nu$ as 
\[\bar \nu(dx) = \sum_{k=-K}^K \bar \nu_k \delta_{kd}(dx),\]
where $K$ is an integer and $\delta$ is the Dirac delta function. Let $P_{i,j}$ denote the approximation of $P(t_i,x_j)$. We use the following implicit-explicit scheme:
\begin{align*}
0&=\frac{P_{i+1,j}-P_{i,j}}{h} + \frac{\sigma^2}{2}
\left(\frac{P_{i,j-1} + P_{i,j+1}-2P_{i,j}}{d^2} - \frac{P_{i,j+1}-P_{i,j-1}}{2d}\right) \\
&+ \sum_{k=-K}^K \left(P_{i+1,j+k}- P_{i+1,j} - (e^{kd}-1) \frac{P_{i+1,j+1}-P_{i+1,j-1}}{2d}\right)\bar\nu_k\\
&+ \min_\vartheta \Bigg\{\frac{\alpha \sigma^2}{2}\left(\vartheta - \frac{P_{i+1,j+1}-P_{i+1,j-1}}{2d}\right)^2 \\ 
&+\frac{1}{\alpha}\sum_{k=-K}^K 
\left( e^{\alpha(P(i+1,j+k)-P(i+1,j) - (e^{kd}-1) \vartheta) }-1 - \alpha(P(i+1,j+k)-P(i+1,j) - (e^{kd}-1) \vartheta) \right)\bar\nu_k\Bigg\}.
\end{align*}
In other words, introducing the notation
\[B_j(P_{i+1}) = 
\sum_{k=-K}^K \left(P_{i+1,j+k}- P_{i+1,j} - (e^{kd}-1) \frac{P_{i+1,j+1}-P_{i+1,j-1}}{2d}\right)\bar\nu_k\]
and 
\begin{align*}
&H_j(P_{i+1},\vartheta)=\frac{\alpha \sigma^2}{2}\left(\vartheta-\frac{P_{i+1,j+1}-P_{i+1,j-1}}{2d}\right)^2 \\
&+\frac{1}{\alpha}\sum_{k=-K}^K 
\left( e^{\alpha(P(i+1,j+k)-P(i+1,j) - (e^{kd}-1) \vartheta) }-1 - \alpha(P(i+1,j+k)-P(i+1,j) -(e^{kd}-1) \vartheta) \right)\bar\nu_k,
\end{align*}
we have for $ j=1,\dots,2M-1$
\begin{multline*}
P_{i,j}\left(1+\frac{\sigma^2h}{d^2}\right) - P_{i,j-1}\left(\frac{\sigma^2h}{2d^2} + \frac{\sigma^2h}{4d}\right)
-P_{i,j+1}\left(\frac{\sigma^2h}{2d^2} - \frac{\sigma^2h}{4d}\right)  \\=P_{i+1,j}+ hB_j(P_{i+1})+ h\min_\vartheta H_j(P_{i+1},\vartheta)
\end{multline*}
with boundary conditions
\[P_{i,2M} = (K-e^{x_{2M}})^+\quad \text{and}\quad P_{i,0} = (K-e^{x_0})^+.\] 

\subsubsection*{Numerical comparison}
In this paragraph we compare numerically the asymptotic formula for
the indifference price of a European put option with
the value obtained by solving the PIDE using a finite difference
scheme. The computational time required to
solve the PIDE with adequate precision is 107 seconds on an iMac
with $2.8$ GHz Intel Core i7 processor (the implementation was done
with Python programming languate, using a single processor core). The parameters of the scheme were
$N = 40$ (number of time steps), $2K = 100$ (number of points to
discretize the L\'evy measure) and $2M = 200$ (number of space steps).

Figure \ref{price.fig}, left graph, plots the price computed with the
two methods
as function of the initial price of the underlying with risk aversion
parameter value $\alpha = 10$.  For comparison, we
also plot the linear part of the price ($\mathbb E^*[H]$), computed using the explicit
formula available in the Merton model. As can be seen from the graph,
the bid-ask spread (that is, twice the difference between the
indifference price and the linear part of the price) for an at the
money option corresponds to about $6\%$ of the option price. This is a
rather high value for the spread, which means that the risk aversion
parameter value which we use is also rather high. 

The graph clearly shows that for the chosen parameter values, which
correspond to a realistic market scenario, the asymptotic
formula is quite precise. To further explore the domain of validity of
the approximation, in the right graph of Figure \ref{price.fig} we
plot the price of an at the money
put option (that is, we take $S_0=1$) as function
of the risk aversion parameter $\alpha$ with a higher resolution. Remark
that the asymptotic formula for the indifference price is linear in
$\alpha$. We see that in this
example the
asymptotic formula reproduces the linear
component of the price with almost no error (for $\alpha=0$), and the
nonlinear component of the price with high precision,
even for relatively large risk aversion values. 
\begin{figure}
\centerline{\includegraphics[width=0.5\textwidth]{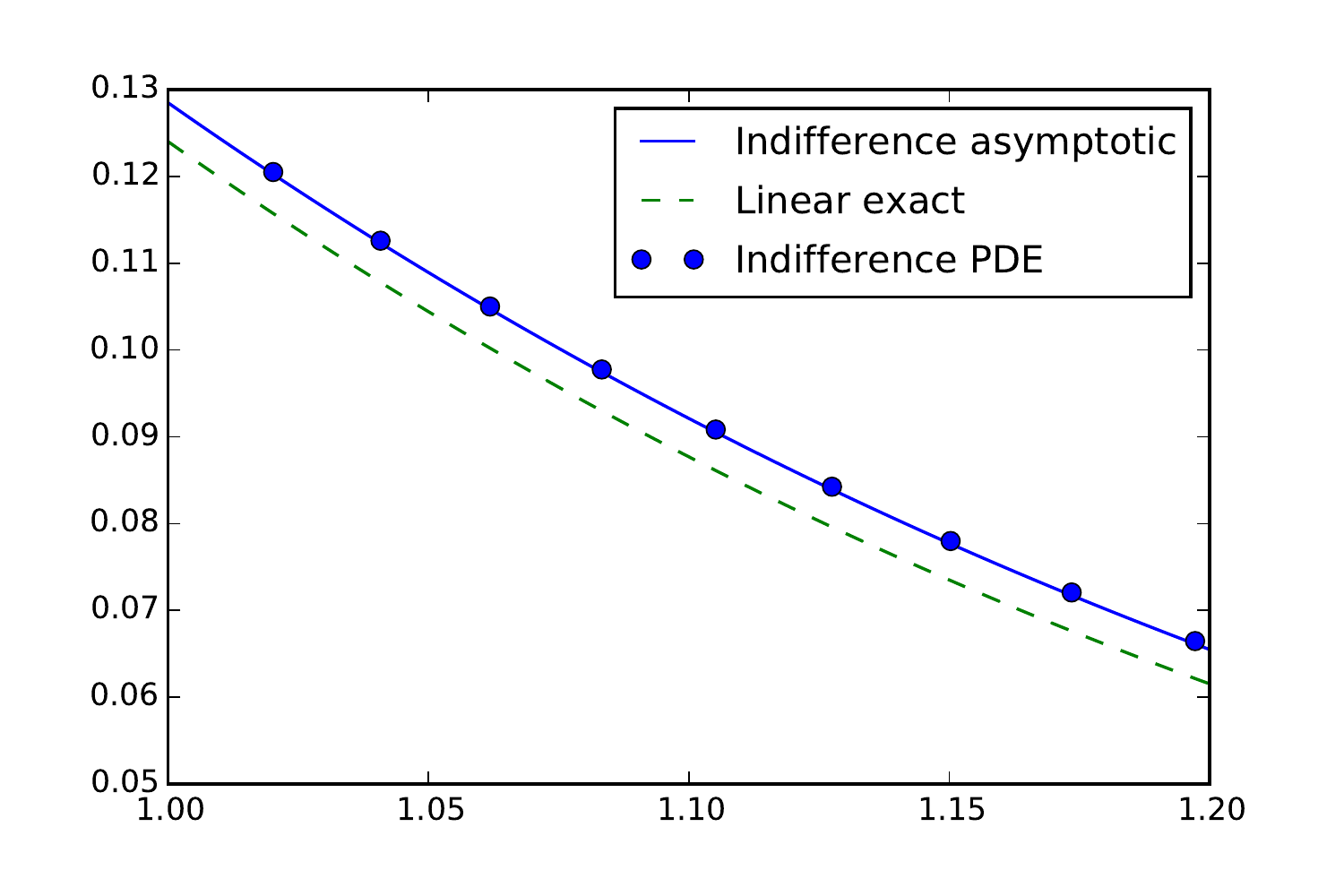}\includegraphics[width=0.5\textwidth]{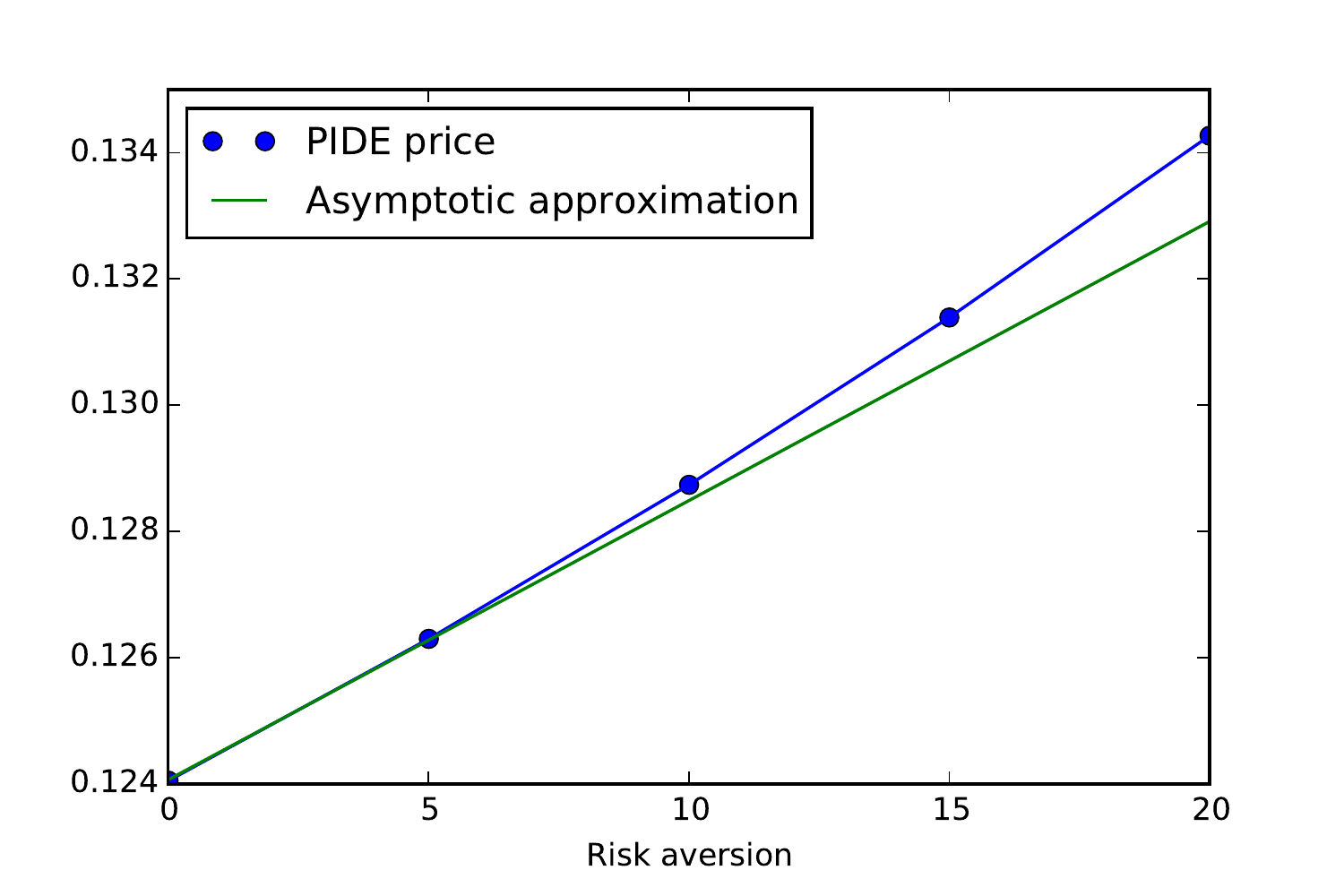}}
\caption{Left: Indifference price computed with PIDE and with the asymptotic formula, in Merton model (log-normal
  jumps), as function of the initial price of the underlying $S_0$ for
  the risk aversion parameter $\alpha=10$. Right: Indifference price
  computed with PIDE and with the asymptotic formula as function of
  the risk aversion parameter $\alpha$ for $S_0=1$. Other model
  parameters: strike $K = 1$, 
  maturity $T = 1$ year, diffusion volatility $\sigma = 0.2$, jump
  intensity $\lambda^M = 5$, average log jump size $-5\%$, log jump size
  standard deviation $10\%$. }
\label{price.fig}
\end{figure}

\section{Bid-ask spread and sensitivity of options to jump risk}
\label{bidask.sec}
As a by-product of the asymptotic formula of Theorem \ref{thm2} we
obtain a simple explicit approximation for the difference between the
{seller's and the buyer's} indifference price of a European option, that is, for the
bid-ask spread:
\begin{equation*}
p_s - p_b \approx \frac{\alpha}{4}\left(m_4 - \frac{m_3^2}{\bar{\sigma}^2}\right) 
\mathbb E^{BS}\left[\int_0^T \left( S_t^2\frac{\partial^2 P_{BS}}{\partial S^2} (t,S_t)\right)^2 dt\right].
\end{equation*}
This spread can be seen as a measure of the effect of market
incompleteness due to jump risk on the price of a specific option,
from the point of view of a specific market agent. It decomposes into
a product of three factors, each representing a specific feature of
our market model:
\begin{itemize}
\item The parameter $\alpha$, which characterizes the
  risk aversion of the economic agent; 
\item The factor $m_4 - \frac{m_3^2}{\bar{\sigma}^2}$ which
    characterizes the specific L\'evy model through its variance, skewness and kurtosis;
\item The expectation of the integral, which characterizes the
  specific option, and only depends on the variance of the price process. 
\end{itemize}
The factor 
\begin{align}
\mathbb E^{BS}\left[\int_0^T \left( S_t^2\frac{\partial^2 P_{BS}}{\partial S^2} (t,S_t)\right)^2 dt\right]\label{exp.eq}
\end{align}
can therefore be seen as a \emph{model-independent} measure of the
sensitivity of a specific European option to jump risk, in the limit
of small jumps. It is therefore interesting to study the dependence of
this measure of jump risk sensitivity on strike and time to maturity. 

\begin{figure}
 \centerline{\includegraphics[width=0.55\textwidth]{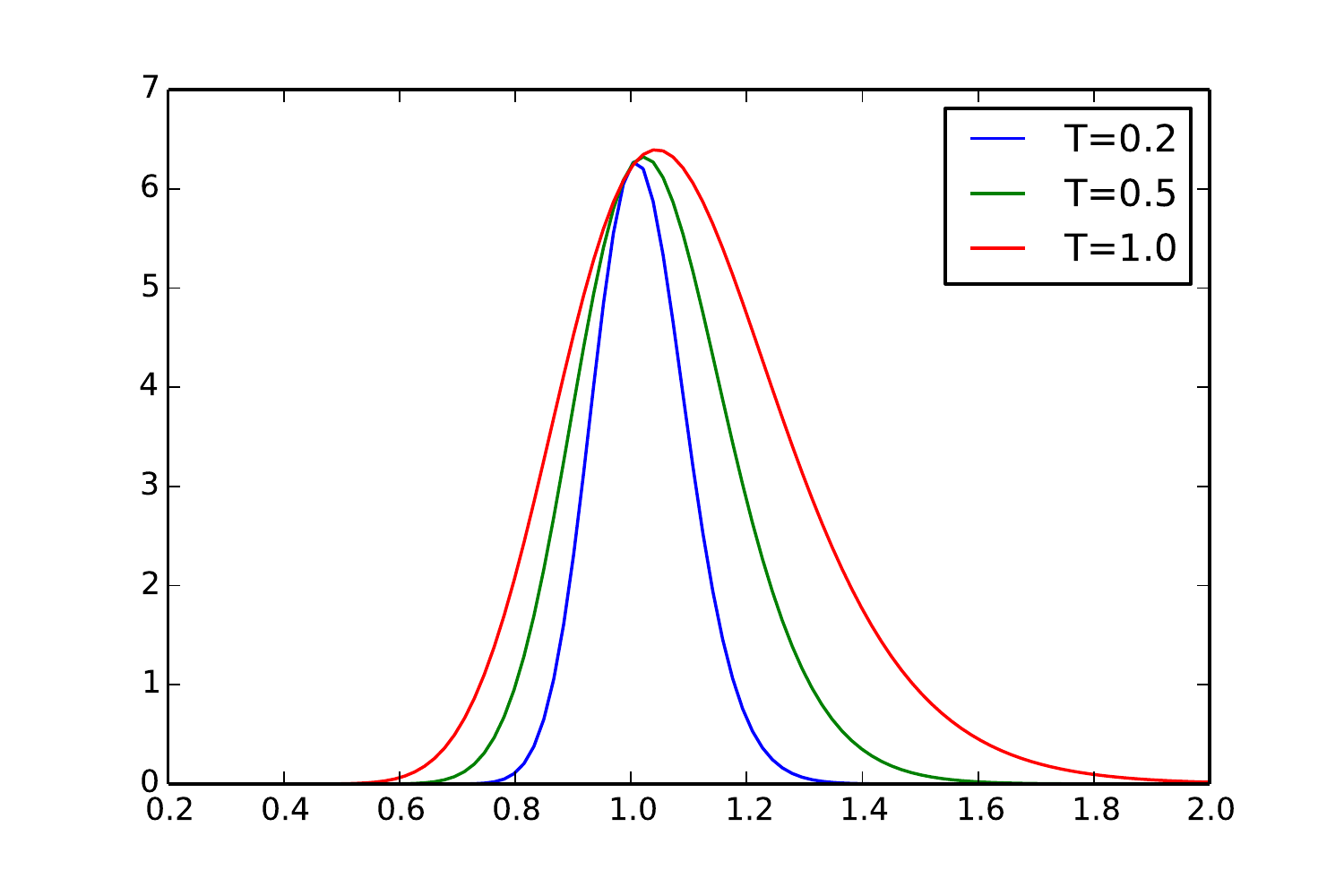}\includegraphics[width=0.55\textwidth]{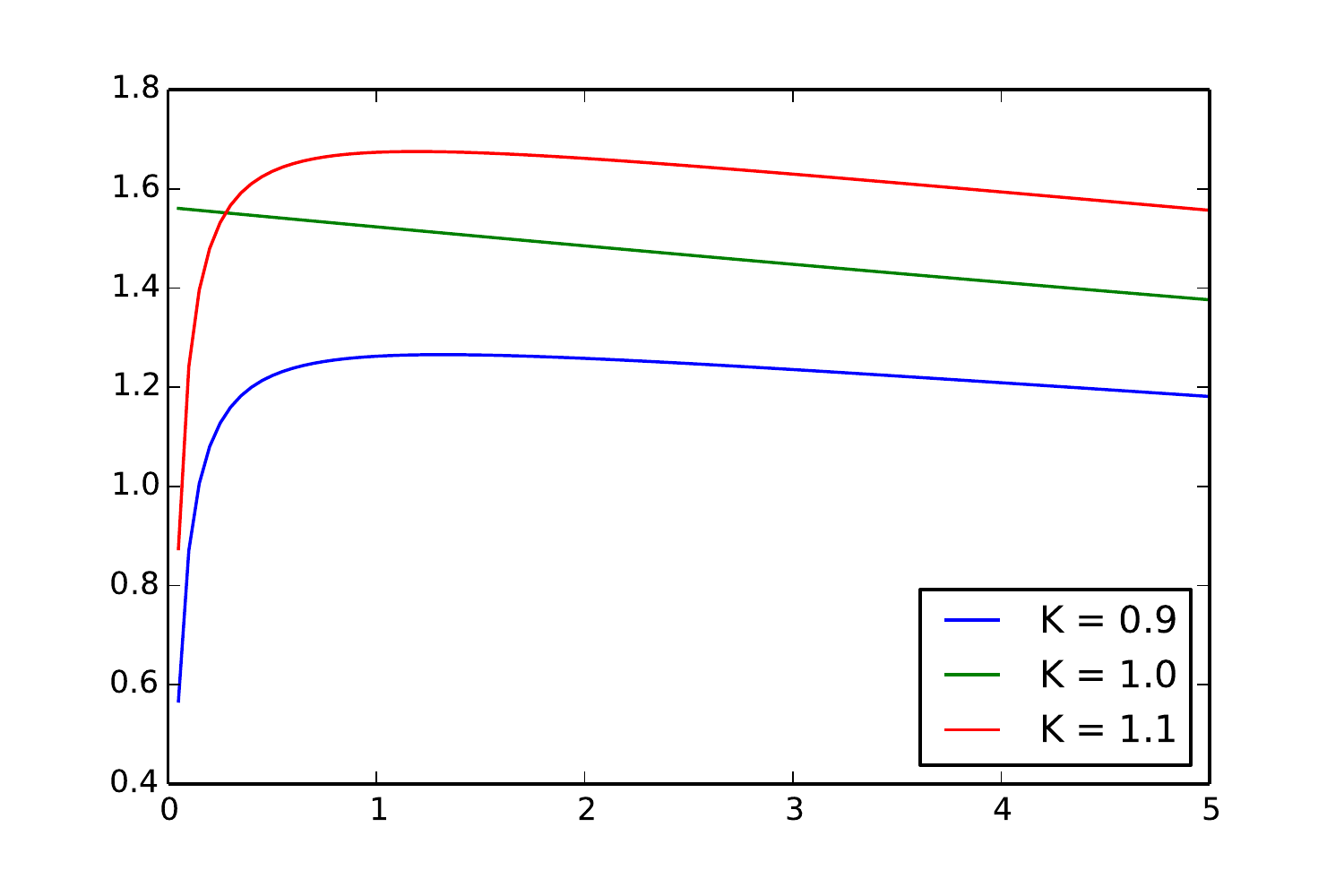}}
\caption{Left: jump risk sensitivity as function of $K$, $S_0=1$, $\bar\sigma =
0.2$. Right: jump risk sensitivity as function of $T$, $S_0=1$, $\bar
\sigma = 0.4$. }
\label{sens.fig}
\end{figure}

Figure \ref{sens.fig} plots the expectation \eqref{exp.eq} as function
of strike (on the left graph) and as function of time to maturity (on
the right graph) {for a European put option}. We see that the sensitivity to jump risk is maximal
for options close to the money, since for far from the money options the exercise
probability and therefore also the price and the spread are very small
(remember that we are interested in sensitivity to small jumps). Note
that actual bid-ask spreads in option markets exhibit similar patterns
with a maximum close to the money (see Figure \ref{market.fig}), although of course actual bid-ask
spreads are influenced by a multitude of factors other than jump
risk. 

In
terms of time to maturity for options which are not at the money, the
sensitivity first grows (because the exercise probability increases)
and then decays for large maturities due to a `central limit theorem'
effect which smoothes out the effect of jumps.  

\begin{figure}
\centerline{\includegraphics[width=0.6\textwidth]{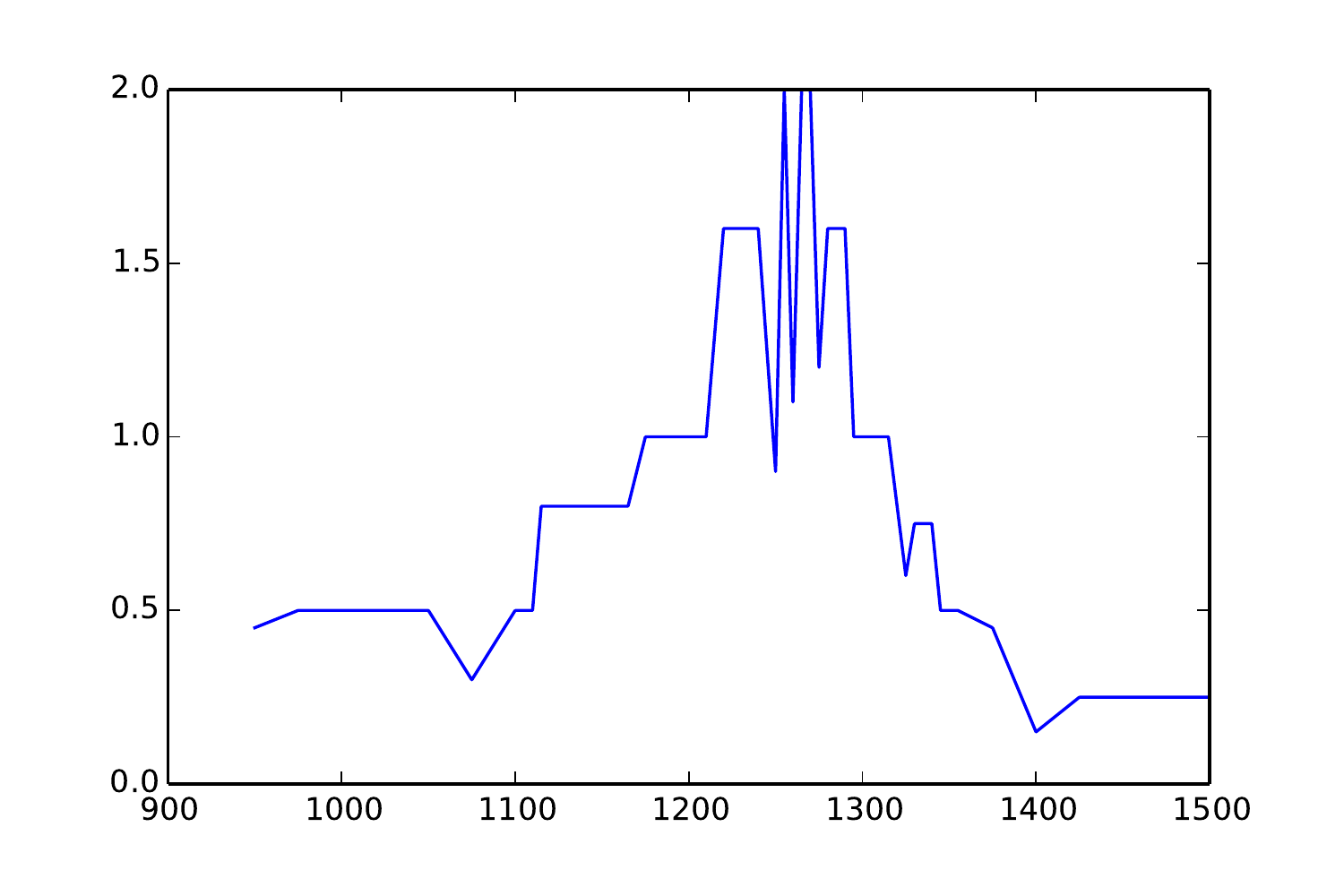}}
\caption{Bid-ask spreads of options on S\&P 500 index, observed on
  Jan 21st, 2006, as function of strike, for time to maturity
  $T=42$ days. The underlying index value was $S_0 = 1261.49$.}
\label{market.fig}
\end{figure}

\appendix

\section{Proof of Lemma \ref{residual.lm}}
Let $P_\lambda(t,S)=
\mathbb{E^*}[h(S\mathcal{E}(X^\lambda)_{T-t})]$. By Proposition 2 in \cite{options}, under the
assumptions of Theorem \ref{thm2}, $P_\lambda(t,S)$ is infinitely
differentiable in $t$ and in $S$, and the assumptions of the Lemma
imply, in particular, that $\left|S\frac{\partial P_\lambda}{\partial S}\right| \leq L$
a.s. for all $t \in [0,T]$.
Using the martingale representation of the option price given in
\cite{options}, we obtain that
\begin{equation*}
M_T^\lambda = \int_0^T \left\lbrace \bar{\vartheta}_t^\lambda  - \frac{\partial P_\lambda}{\partial S} \right\rbrace \sigma S_t^\lambda dW_t
+ \int_0^T \int_{\mathbb{R}} \lbrace z \bar{\vartheta}_t^\lambda  S_{t-}^\lambda   - P_\lambda(t,S_{t-}^\lambda (1+z)) + P_\lambda(t,S_{t-}^\lambda ) \rbrace \tilde{J}_{X^\lambda}(dt\ dz)
\end{equation*}
and that the quadratic hedging strategy is given by
\begin{equation*}
\bar{\vartheta}_t^\lambda = \frac{\sigma^2 \frac{\partial
    P_\lambda}{\partial S}+ \frac{1}{S_t^\lambda}\int_{\mathbb
    R}z(P_\lambda(t,S_t^\lambda(1+z)) -
  P_\lambda(t,S_t^\lambda))\nu_\lambda(dz)}{\sigma^2+\int_{\mathbb
    R}z^2 \nu_\lambda(dz)} = \frac{\partial P_\lambda}{\partial S} +
\frac{1}{\bar \sigma^2 S^\lambda_t } \int_{\mathbb R} z \Xi^\lambda_t(z) \nu_\lambda(dz),
\end{equation*} 
where we denote
$$
\Xi_t^\lambda(z) =  P_\lambda(t,S_t^\lambda(1+z)) - P_\lambda(t,S_t^\lambda) - zS_t^\lambda  \frac{\partial P_\lambda}{\partial S}.
$$

By the Burkholder-Davis-Gundy inequality expressed in predictable
terms \cite{BDG}, for $q\geq 2$, there exist $c_q, C_q >0$ such that:
\begin{equation}
\mathbb{E} [(\sup_{0\leq t\leq T} |M_t|)^q]
\leq C_q \mathbb{E} [\langle M\rangle _T^{\frac{q}{2}} +  |x|^q \star
\nu_T^{M}]\label{bdg.eq}
\end{equation}
where $\nu_T^{M}$ is the compensator of the jump measure of
the process $M$. The quantities appearing in the right-hand side are
explicitly given by
\begin{align*}
&\langle M^\lambda \rangle_T  = \int_0^T \left\lbrace \bar{\vartheta}_t^\lambda  - \frac{\partial P_\lambda}{\partial S} \right\rbrace^2 \sigma^2  (S_t^\lambda)^2 dt 
+ \int_0^T \int_{\mathbb{R}} \lbrace z \bar{\vartheta}_t^\lambda  S_{t}^\lambda   - P_\lambda(t,S_{t-}^\lambda (1+z)) + P_\lambda(t,S_{t-}^\lambda ) \rbrace^2 \nu_\lambda(dz)dt,\\
&|x|^q \star \nu_T^{M^\lambda} 
= \int_0^T \int_{\mathbb{R}} \left| z \bar{\vartheta}_t^\lambda  S_{t}^\lambda   - P_\lambda(t,S_{t-}^\lambda (1+z)) + P_\lambda(t,S_{t-}^\lambda ) \right|^q \nu_\lambda(dz)dt.
\end{align*}
Substituting the expression for $\bar\vartheta^\lambda$ into the first
of the above
equalities, we further obtain:
\begin{align}
\langle M^\lambda \rangle_T &= \frac{\sigma^2}{\bar \sigma^4} \int_0^T
dt\left(\int_{\mathbb R} z \Xi^\lambda_t(z)
  \nu_\lambda(dz)\right)^2 + \int_0^T \int_\mathbb R \left(
  \frac{z}{\bar \sigma^2} \int_{\mathbb R} z \Xi^\lambda_t(z) \nu_\lambda(dz) - \Xi^\lambda_t(z)\right)^2
\nu_\lambda(dz) dt\notag\\
& = - \frac{1}{\bar \sigma^2} \int_0^T
dt\left(\int_{\mathbb R} z \Xi^\lambda_t(z)
  \nu_\lambda(dz)\right)^2 + \int_0^T \int_\mathbb R \left( \Xi^\lambda_t(z)\right)^2
\nu_\lambda(dz) dt \leq \int_0^T \int_\mathbb R \left( \Xi^\lambda_t(z)\right)^2
\nu_\lambda(dz) dt. \label{bracket.eq}
\end{align}

Our first goal is to estimate $\langle M^\lambda\rangle_T$. To this
end, we fix $\eta\in(0,T)$ and estimate separately $\langle
M^\lambda\rangle_{T-\eta}$ and $\langle M^\lambda\rangle_T - \langle
M^\lambda\rangle_{T-\eta}$. 

We proceed with the estimation of $\langle
M^\lambda\rangle_{T-\eta}$. By Taylor--Lagrange expansion
\begin{equation*}
\Xi^\lambda_t(z) = z^2(S_t^\lambda)^2 \int_0^1 \frac{\partial^2 P_\lambda}{\partial S^2}(t, S_t^\lambda(1+\theta z))(1-\theta) d\theta.
\end{equation*}
Note that $\forall z>-1$, $\int_0^1 \frac{1-\theta}{1+\theta z} d\theta = \frac{(1+z)\log(1+z)-z}{z^2} \leq 1$. Thus denoting $\Gamma_\lambda(t,S) = \frac{\partial^2 P_\lambda}{\partial S^2}(t,S)$
\begin{equation*}
| \Xi_t^\lambda(z)|  \leq {z^2}S_t^\lambda \max_{S}  S |\Gamma_\lambda(t,S)|.
\end{equation*}
We conclude that
\begin{equation}
\langle M^\lambda \rangle_{T-\eta}
\leq {\lambda^2} \int_\mathbb{R} z^4 \nu(dz)
\int_0^{T-\eta} (S_t^\lambda)^2\left(\max_{S}  S |\Gamma_\lambda(t,S)|\right)^2dt.\label{gammaest}
\end{equation}

Furthermore, the term $\langle M^\lambda \rangle_{T} - \langle
M^\lambda \rangle_{T-\eta}$ is controlled by the {
option's delta.}
Indeed, as 
$\left|S\frac{\partial P_\lambda}{\partial S}\right| \leq L$, we
deduce that 
\begin{equation*}
|P_\lambda(t,S_t^\lambda(1+z))-P_\lambda(t,S_t^\lambda)| = \left|\int_0^z S_t^\lambda\frac{\partial P_\lambda}{\partial S} (t,S_t^\lambda(1+u))du \right| \leq \frac{L|z|}{1-\delta},
\end{equation*}
so that finally
\begin{align}
&\langle M^\lambda \rangle_{T} - \langle
M^\lambda \rangle_{T-\eta}\leq \int_{T-\eta}^T \int_\mathbb R \left( \Xi^\lambda_t(z)\right)^2
\nu_\lambda(dz) dt \notag\\ &\leq 2 \int_{T-\eta}^T \int_\mathbb R
|P_\lambda(t,S_t^\lambda(1+z))-P_\lambda(t,S_t^\lambda)|^2
\nu_\lambda(dz) dt + 2 \int_{T-\eta}^T \int_\mathbb R z^2
(S^\lambda_t)^2 \left|\frac{\partial P_\lambda}{\partial S}\right|^2
\nu_\lambda(dz) dt\notag\\
&\leq 2 \eta L^2 \int_{\mathbb R} z^2\nu(dz)\left(
  \frac{1}{(1-\delta)^2} + 1\right) .\label{deltaest}
\end{align}

Using \eqref{gammaest}, \eqref{deltaest} and the {
gamma} estimate of
Lemma \ref{gamma.lm}, we conclude that there exist two constants $A$ and
$B$ such that 
\begin{equation*}
\langle M^\lambda \rangle_{T} \leq \left\lbrace \begin{array}{cc}
\lambda^2A \sup_{0\leq t\leq T} (S^\lambda_t)^2 \ (\log T - \log \eta) + \eta B  & \sigma>0, \\
\lambda^2A  \sup_{0\leq t\leq T} (S^\lambda_t)^2 \left[(\log T - \log \eta) +\lambda^{2(\frac{2}{\beta}-1)}\left(\frac{1}{\eta^{\frac{2}{\beta}-1}}-\frac{1}{T^{\frac{2}{\beta}-1}} \right)\right] +\eta B  & \sigma=0.
\end{array}\right.
\end{equation*}

Taking $\eta = \lambda^2$ and assuming $\lambda \leq \sqrt{T}\wedge
\frac{1}{2}$, yields, for a constant $C<\infty$, 
$$
\langle M^\lambda\rangle_T \leq \lambda^2 C\left\{1-\sup_{0\leq t\leq
    T} (S^\lambda_t)^2\log \lambda \right\}.
$$
Since $S^\lambda$ is a martingale, by the BDG inequality we have, for
all $q\geq 2$, 
\begin{align*}
\mathbb E[|S^\lambda_T-S^\lambda_0|^q] &\leq \mathbb E[\sup_{0\leq t\leq
  T}|S^\lambda_t - S^\lambda_0|^q]\\ &\leq C_q \mathbb E\left[\left(\sigma^2 +
    \int_{\mathbb R} x^2 \nu(dx)\right)^{q/2}\left(\int_0^T S_t^2
    dt\right)^{q/2 } + \lambda^{q-2}\int_\mathbb R z^q \nu(dz)
  \int_0^T S_t^q  dt\right],
\end{align*}
which means that $\mathbb E[\sup_{0\leq t\leq T}|S^\lambda_t|^q]$ is
bounded uniformly on $\lambda$, and therefore,
$$
\mathbb E [\langle M^\lambda\rangle^{q/2}_T]  = O\left(\lambda^q
  \left(\log \frac{1}{\lambda}\right)^{\frac{q}{2}}\right).
$$

Lastly, it remains to estimate the second term in the BGD inequality \eqref{bdg.eq}.
In a similar manner as for the quadratic variation of $M^\lambda$, we
obtain for every $q>2$ and for some constant $C<\infty$ (which may
change from line to line), 
\begin{align*}
|x|^q &\star \nu_T^{M^\lambda,P} = \int_0^T \int_\mathbb R
\left|\frac{z}{\bar \sigma^2} \int_\mathbb R
  z\Xi^\lambda_t(z)\nu_\lambda(dz)-\Xi^\lambda_t(z)\right|^q\nu_\lambda(dz)
dt \\ &\leq 2^{q-1}  \bar \sigma^{-2q} \lambda^{q-2}\int_\mathbb R
\left|{z}\right|^q\nu(dz)\int_0^T \left|\int_\mathbb R
  z\Xi^\lambda_t(z)\nu_\lambda(dz)\right|^q
dt + 2^{q-1}\int_0^T \int_\mathbb R
\left|\Xi^\lambda_t(z)\right|^q\nu_\lambda(dz)
dt \\
&\leq C \lambda^{2q-2}\int_0^{T-\eta} |S_t^\lambda \max_{S}  S
|\Gamma_\lambda(t,S)||^q dt   + C \eta \lambda^{q-2}\\
&\leq \left\lbrace \begin{array}{cc}
\lambda^{2q-2} C\sup_{0\leq t\leq T} (S^\lambda_t)^q \left(\frac{1}{\eta^{\frac{q}{2}-1}}-\frac{1}{T^{\frac{q}{2}-1}} \right) + \eta \lambda^{q-2} C  & \sigma>0, \\
\lambda^{2q-2} C \sup_{0\leq t\leq T} (S^\lambda_t)^q\left[\left(\frac{1}{\eta^{\frac{q}{2}-1}}-\frac{1}{T^{\frac{q}{2}-1}} \right) +{\lambda^{2(\frac{2}{\beta}-1)}}\left(\frac{1}{\eta^{\frac{2}{\beta}-1}}-\frac{1}{T^{\frac{2}{\beta}-1}} \right)\right] +\eta \lambda^{q-2} C  & \sigma=0.
\end{array}\right.
\end{align*}
Choosing once again $\eta = \lambda^2$ leads to,
\begin{equation*}
|x|^q \star \nu_T^{M^\lambda,P}  \leq C \lambda^q (1+\sup_{0\leq t\leq T} (S^\lambda_t)^q )
\end{equation*}
for some constant $C$, and therefore, 
$$
\mathbb E[|x|^q \star \nu_T^{M^\lambda,P}] = O(\lambda^q)
$$
as $\lambda \to 0$. 
The proof of Lemma \ref{residual.lm} is complete.

\begin{lem}[Estimation of the gamma]${}$\label{gamma.lm}
\begin{itemize}
\item Let $\sigma>0$. Then there exists $C<\infty$ such that 
$$
\max_{S}  S |\Gamma_\lambda(t,S)| \leq \frac{C}{\sqrt{T-t}}.
$$
\item Let $\sigma =0$ and $\liminf\limits_{r\downarrow 0}
  \frac{\int_{[-r,r]} x^2 \nu(dx)}{r^{2-\beta}} > 0$ for some $\beta
  \in (0,2)$. Then there exists $C<\infty$ such that
$$
\max_{S}  S |\Gamma_\lambda(t,S)|\leq C\left\{\frac{1}{\sqrt{T-t}} + \frac{\lambda^{\frac{2}{\beta}-1}}{(T-t)^{\frac{1}{\beta}}}\right\}.
$$
\end{itemize}
\end{lem}

\begin{proof}
Let 
$$
\text{Call}^K_\lambda(t,S) = \mathbb E[\left(S\mathcal E(X^\lambda)_{T-t}-K\right)^+].
$$
Under the assumptions of Theorem \ref{thm2}, 
$$
h(S) = h(0) + h'(0)S + \int_{(0,\infty]}(S-K)^+\mu(dK),
$$
where $\mu$ is a finite measure on $\mathbb R^+$, defined by
$\mu((a,b]) = h'(b) - h'(a)$. Since $\text{Call}^K_\lambda(t,S) \leq
S$, by Fubini's theorem,
$$
P_\lambda(t,S) = h(0) + h'(0)S + \int_{(0,\infty]}
\text{Call}^K_\lambda(t,S)\mu(dK). 
$$
Moreover, since $\frac{\partial \text{Call}^K_\lambda(t,S)}{\partial
  S}$ is bounded and $\mu$ is a finite measure, the dominated convergence theorem yields
$$
\frac{\partial P_\lambda(t,S)}{\partial S} = h'(0) + \int_{(0,\infty]}
\frac{\partial \text{Call}^K_\lambda(t,S)}{\partial
  S}\mu(dK). 
$$

Using the Fourier transform representation for call option price in
exponential Lévy models \cite{tankov.pplnmf}, we get the following identity:
\begin{align*}
\frac{\partial^2 \text{Call}^K_\lambda}{\partial S^2}(t,S) 
&= \frac{1}{2\pi} \int_\mathbb{R} K^{iu+1}S^{-iu-2}
\Phi_{T-t}^\lambda(-u)du\\
& =\frac{1}{2\pi} \int_\mathbb{R} K^{iu}S^{-iu-1}
\Phi_{T-t}^\lambda(-u-i)du
\end{align*}
where $\Phi_t^\lambda(u) = \mathbb{E}[e^{iu\log \mathcal{E}(X^\lambda)_t}] = e^{t\psi_\lambda(u)}$ and
$\psi_\lambda(u) = -\frac{\sigma^2}{2}(u^2+iu)+\int (e^{iu\log(1+z)} -1
-iuz)\nu_\lambda(dz)$.
Therefore,
$$
\left|\frac{\partial^2 \text{Call}^K_\lambda}{\partial S^2}(t,S)
\right| \leq \frac{1}{2\pi S} \int_{\mathbb R} \left|\Phi^\lambda_{T-t}(-u-i)\right|du,
$$
and the dominated convergence theorem yields
$$
S\left|\Gamma_\lambda(t,S) \right|\leq
\frac{C}{2\pi }  \int_{\mathbb R}
\left|\Phi^\lambda_{T-t}(-u-i)\right|du = \frac{C}{2\pi }  \int_{\mathbb R}
e^{(T-t)\Re \psi_\lambda(u-i)}du
$$
where $C = \int_{(0,\infty]} |d\mu|$ and $\Re\psi_\lambda(u-i) = -\frac{\sigma^2u^2}{2} + \int \{(1+x)(\cos(u\log(1+x))-1)\} \nu_\lambda(dx)$.

Let us study separately the cases $\sigma >0$ and $\sigma =0$. When $\sigma>0$, we directly get
\begin{equation*}
\Re\psi_\lambda(u-i) \leq  -\frac{\sigma^2u^2}{2}
\end{equation*}
which leads to
\begin{equation*}
S \left| \Gamma_\lambda(t,S)\right|
\leq \frac{C}{\sigma\sqrt{2\pi(T-t)}}.
\end{equation*}

When $\sigma=0$, using \eqref{boundlambda.eq}, we get,
\begin{align*}
S \left| \Gamma_\lambda(t,S) \right| &\leq \frac{C}{2\pi}\left( \int_{|u|< \ell/\lambda} e^{-c(T-t)u^2}du 
+ \int_{|u|\geq \ell/\lambda} e^{-\frac{c|u|^\beta}{\lambda^{2-\beta}}
  (T-t)}du\right)\\
&\leq \frac{C}{2\pi}\left( \int_{\mathbb R} e^{-c(T-t)u^2}du 
+ \int_{\mathbb R} e^{-\frac{c|u|^\beta}{\lambda^{2-\beta}}
  (T-t)}du\right)\\
& = \tilde C\left\{\frac{1}{\sqrt{T-t}} + \frac{\lambda^{\frac{2}{\beta}-1}}{(T-t)^{\frac{1}{\beta}}}\right\},
\end{align*}
where $\tilde C$ is a constant. 
\end{proof}

\begin{lem}
Assume that $\sigma =0$ and $\liminf\limits_{r\downarrow 0}
  \frac{\int_{[-r,r]} x^2 \nu(dx)}{r^{2-\beta}} > 0$ for some $\beta
  \in (0,2)$. Then, for every $\ell>0$ and $v\in \mathbb R$, there
  exist $c>0$ and $C<\infty$ such that for $u\in \mathbb R$
\begin{equation}
\Re \psi_\lambda(u+iv)
\leq  \left\lbrace \begin{array}{cc}
C -c u^2&|u\lambda |\leq \ell\\
C -c |u|^\beta \lambda^{\beta-2} &|u\lambda| \geq \ell.
\end{array}\right.\label{boundlambda.eq}
\end{equation}

\end{lem}
\begin{proof}
Observe first that
\begin{align*}
\Re &\psi_\lambda(u+iv) = \lambda^{-2}\int ((1+\lambda x)^{-v} \cos(u(1+\lambda x))
- 1 + v \lambda x)\nu(dx) \\
& = \lambda^{-2}\int (1+\lambda x)^{-v} (\cos(u(1+\lambda x))-1)
\nu(dx) + v(v+1)\int x^2 \nu(dx)\int_0^1 (1+\theta\lambda
x)^{-v-2} (1-\theta) d\theta. 
\end{align*}
Since $X$ has jumps bounded by $\delta$ and $\lambda \leq 1$, there
exists $C<\infty$ and $c>0$ (depending on $v$) such that for all
$x$ in the support of $\nu$, 
$$
|v(v+1)\int x^2 \nu(dx)\int_0^1 (1+\theta\lambda
x)^{-v-2} (1-\theta) d\theta |\leq C\quad \text{and}\quad |1+\lambda
x|^{-v}\geq c.
$$ 
Then, using $1-\cos(x)=2(\sin \frac{x}{2})^2 \geq 2(\frac{x}{\pi})^2$ for $|x|\leq \pi$, we get:
\begin{align*}
\Re \psi_\lambda(u+iv)\leq C -c \int_{|u\log(1+\lambda x)|\leq \pi} 
\frac{u^2 (\log(1+\lambda x))^2}{\lambda^2} \nu(dx),
\end{align*}
but since $(\log(1+x))^2 \geq x^2 (\log 2)^2$ for $|x|\leq 1$ we have,
for a different $c>0$, 
\begin{equation*}
\Re \psi_\lambda(u+iv) \leq C - c  \int_{|u\log(1+\lambda x)|\leq \pi} 
{ (ux)^2 } 1_{|\lambda x| \leq 1}\nu(dx).
\end{equation*}
Once again, by the bound on the jumps of $X$,  $|\log(1+x)| \leq \frac{|x|}{\log
  (1+\delta)}$ on the support of $\nu$ and we also have,
for a different $c>0$, 
\begin{equation*}
\Re \psi_\lambda(u+iv) \leq C-c
\int_{|x|\leq \frac{\pi}{\lambda|u|\log(1+\delta) }\wedge 1} u^2 x^2 \nu(dx)	.
\end{equation*}
Under the assumption $\liminf\limits_{r\downarrow 0} 
\frac{\int_{[-r,r]} x^2 \nu(dx)}{r^{2-\beta}}>~0$ for some $\beta \in
(0,2)$, there exist $r_0>0$ and $c_0>0$ such that for all $r<r_0$, 
$
\int_{[-r,r]} x^2 \nu(dx) \geq c_0 r^{2-\beta}.
$ This implies that one can find constants $\ell>0$ and yet another $c>0$ such
that 
\begin{equation*}
\Re \psi_\lambda(u+iv)
\leq  \left\lbrace \begin{array}{cc}
C-c u^2&|u\lambda |\leq \ell\\
C-c |u|^\beta \lambda^{\beta-2} &|u\lambda| \geq \ell
\end{array}\right.
\end{equation*}
Now, by changing $c$ and $C$ this inequality can be
shown to be true for arbitrary $\ell$.  
\end{proof}

\section{Proof of Lemma \ref{nonlinear.lm}}
Using the notation of the proof of Lemma \ref{residual.lm}, we have 
$$
\mathbb{E} \left[\left(\int_0^T \bar\vartheta^\lambda_t dS^\lambda_t
    -h(S^\lambda_T)\right)^2\right] = \mathbb E[(M^\lambda_T)^2] =
\mathbb E[\langle M^\lambda\rangle_T],
$$
where $\langle M^\lambda\rangle_T$ was computed in
\eqref{bracket.eq}. 
From the Fourier transform
formula for the call option price 
$$
\text{Call}^K_\lambda(t,S) = \frac{1}{2\pi} \int_{\mathbb R} \frac{K^{iu+1-R} S_t^{-iu+ R}\Phi^\lambda_{T-t}(-u-iR) }{(R-iu)(R-1-iu)}
du,
$$
we deduce that 
\begin{multline*}
\text{Call}^K_\lambda(t,S(1+z)) - \text{Call}^K_\lambda(t,S) - zS
\frac{\partial \text{Call}^K_\lambda(t,S)}{\partial S} \\= \frac{1}{2\pi} \int_{\mathbb R}
{K^{iu+1-R} S_t^{-iu+ R}\Phi^\lambda_{T-t}(-u-iR) }\frac{(1+z)^{-iu+R}
- 1 + (iu-R)z}{(R-iu)(R-1-iu)}du.
\end{multline*}
Since the fraction under the integral sign is analytic for $z>-1$, we
can choose $R=1$ in this formula, obtaining
\begin{multline*}
\text{Call}^K_\lambda(t,S(1+z)) - \text{Call}^K_\lambda(t,S) - zS
\frac{\partial \text{Call}^K_\lambda(t,S)}{\partial S}\\= \frac{1}{2\pi} \int_{\mathbb R}
{K^{iu} S_t^{-iu+ 1}\Phi^\lambda_{T-t}(-u-i) }\frac{(1+z)^{-iu+1}
- 1 + (iu-1)z}{iu(iu-1)}du.
\end{multline*}
and therefore,
$$
\Xi_t^\lambda(z) = \frac{1}{2\pi} \int_{(0,\infty)}\mu(dK) \int_{\mathbb R}
{K^{iu} S_t^{-iu+ 1}\Phi^\lambda_{T-t}(-u-i) }\frac{(1+z)^{-iu+1}
- 1 + (iu-1)z}{iu(iu-1)}du.
$$
By Fubini's theorem, the expression
under the time integral in \eqref{bracket.eq} {
equals  }
\begin{align*}
&\frac{1}{4\pi^2}\int_{(0,\infty)}\mu(dK) \int_{(0,\infty)}\mu(d\bar K)\int_{\mathbb R} dv\int_{\mathbb  R} du \,
K^{iu}\bar K^{iv}\\ &\qquad \qquad \times\frac{\Phi^\lambda_{T-t}(-u-i)\Phi^\lambda_{T-t}(-v-i)
  \Phi^\lambda_t(-u-v-2i)}{-uv(iu-1)(iv-1)}\left\{-\frac{a_\lambda(u)a_\lambda(v)}{\bar
  \sigma^2} + b_\lambda(u,v) \right\},
\end{align*}
where
\begin{align*}
&a_\lambda(u) = \int_{\mathbb R} \nu_\lambda(dz) z
\{(1+z)^{-iu+1}-1+(iu-1)z\} = \lambda^{-1} \int_{\mathbb R} \nu(dz) z
\{(1+\lambda z)^{-iu+1}-1+(iu-1)\lambda z\} \\&\text{and}\quad b_\lambda(u,v) =
\int_{\mathbb R} \nu_\lambda(dz)
\{(1+z)^{-iu+1}-1+(iu-1)z\}\{(1+z)^{-iv+1}-1+(iv-1)z\}\\
& \qquad \qquad \qquad = \lambda^{-2}\int_{\mathbb R} \nu(dz)
\{(1+\lambda z)^{-iu+1}-1+(iu-1)\lambda z\}\{(1+\lambda z)^{-iv+1}-1+(iv-1)\lambda
z\}
\end{align*}
From the explicit form of $\Phi^\lambda$,
\begin{align*}
\Phi^\lambda_t(u) = e^{t\psi_\lambda(u)},\qquad \psi_\lambda(u) &=
-\frac{\sigma^2}{2} (u^2+iu) + \lambda^{-2}\int_{\mathbb R}((1+\lambda
  z)^{iu}-1-iu\lambda z)\nu(dz)\\
& = -\frac{\sigma^2}{2} (u^2+iu) )- (u^2+iu)\int_{\mathbb
  R} z^2 \nu(dz)\int_0^1 ((1+\lambda
  z\theta)^{iu-2}(1-\theta) d\theta,
\end{align*}
we deduce, using the fact that $\nu$ has bounded support, that for
every $u\in \mathbb C$, 
$$
\psi^\lambda(u) \to \psi^0(u) = -\frac{u^2+iu}{2}\left(\sigma^2 +
  \int_{\mathbb R} z^2 \nu(dz)\right)\quad \text{as}\quad \lambda\to
0. 
$$
On the other hand, since 
$$
{(1+\lambda
  z)^{-iu+1}-1+(iu-1)\lambda z}={{iu (iu-1)}\lambda^2 z^2} \int_0^1 (1+\lambda\theta
 z)^{-iu-1} (1-\theta)d\theta,
$$
we get that 
$$
\lambda^{-1} a_\lambda(u)\to \frac{iu(iu-1)}{2}\int_{\mathbb R}z^3 \nu(dz) \quad
\text{and}\quad \lambda^{-2} b_\lambda(u,v)\to -\frac{uv(iu-1)(iv-1)}{4}\int_{\mathbb R}z^4 \nu(dz)
$$
as $\lambda \to 0$. 
Therefore, provided that we can find an integrable bound to apply the
dominated convergence theorem, $\lambda^{-2}\langle
M^\lambda\rangle_T$ converges to 
\begin{align*}
& \frac{1}{4}\left(m_4 -
  \frac{m_3^2}{\bar\sigma^2}\right)\frac{1}{4\pi^2}  \int_{(0,\infty)}\mu(dK) \int_{(0,\infty)}\mu(d\bar K)\int_0^T dt \int_{\mathbb R} dv\int_{\mathbb  R} du \,
K^{iu}\bar K^{iv}\\& \qquad \qquad \times\Phi^0_{T-t}(-u-i)\Phi^0_{T-t}(-v-i) \Phi^0_t(-u-v-2i)\\
& = \frac{1}{4}\left(m_4 - \frac{m_3^2}{\bar\sigma^2}\right) \mathbb E^{BS}\left[\int_0^T \left(S_t^2
    \frac{\partial P_{BS}(t,S_t)}{\partial S^2}\right)^2 dt\right]
\end{align*}
as $\lambda\to 0$. 

We first consider the case $\sigma>0$. 
Remark first that by the bound on the jumps of $X$, 
$$
\left|\int_0^1 (1+\lambda\theta
 z)^{-iu-1} (1-\theta)d\theta\right| \leq \frac{1}{2(1-\delta)},
$$
so that it remains to find an integrable (in $u$, $v$ and $t$) bound for 
$$
\Phi^\lambda_{T-t}(-u-i)\Phi^\lambda_{T-t}(-v-i) \Phi^\lambda_t(-u-v-2i) 
$$
However, in this case, 
\begin{align*}
&\left| \Phi^\lambda_{T-t}(-u-i)\Phi^\lambda_{T-t}(-v-i)
  \Phi^\lambda_t(-u-v) \right|\leq e^{-\frac{1}{2}(T-t)\sigma^2
  (u^2+v^2) -\frac{1}{2}t\sigma^2 (u+v)^2 + t \bar \sigma^2},
\end{align*}
which is integrable since 
$$
\int_{\mathbb R} du \int_{\mathbb R} dv e^{-\frac{1}{2}(T-t)\sigma^2
  (u^2+v^2) -\frac{1}{2}t\sigma^2 (u+v)^2}  = \int_{\mathbb  R} du\, e^{-\frac{1}{2}T\sigma^2
  u^2} \int_{\mathbb R} dv \, e^{ -\frac{1}{2}T \sigma^2
  v^2(1-\frac{t^2}{T^2}) } = \frac{2\pi}{\sigma^2 \sqrt{T^2-t^2}}. 
$$

Let us now consider the case $\sigma = 0$. We shall use the bound \eqref{boundlambda.eq}. 
In addition, 
\begin{align*}
|b_\lambda(u,v)|&\leq \lambda^{-2}\left(\int \nu(dz)|(1+\lambda
  z)^{-iu+1}-1+(iu-1)\lambda z|^2\right)^{\frac{1}{2}} \\ &\times \left(\int \nu(dz)|(1+\lambda
  z)^{-iv+1}-1+(iv-1)\lambda z|^2\right)^{\frac{1}{2}}
\end{align*}
and it is easy to show, using arguments similar to those used to prove
the bound \eqref{boundlambda.eq} that for some constant $C<\infty$, 
\begin{align*}
&\int \nu(dz)|(1+\lambda z)^{-iu+1}-1+(iu-1)\lambda z|^2 \leq C \lambda^4 u^2 (u^2+1) 1_{|\lambda u|
  \leq \ell} +C \lambda^2 u^2 1_{|\lambda u
  |>\ell},
\end{align*}
where the constant $\ell$ may be taken the same as in the bound
\eqref{boundlambda.eq}. Similarly, by the Cauchy-Schwarz inequality, 
\begin{align*}
|a_\lambda(u)a_\lambda(v)|&\leq \bar \sigma^2\lambda^{-2}\left(\int \nu(dz)|(1+\lambda
  z)^{-iu+1}-1+(iu-1)\lambda z|^2\right)^{\frac{1}{2}}\\ &\times \left(\int \nu(dz)|(1+\lambda
  z)^{-iv+1}-1+(iv-1)\lambda z|^2\right)^{\frac{1}{2}},
\end{align*}
so that to complete the proof it suffices to study the integral 
\begin{align}
\int_{\mathbb R} dv\int_{\mathbb  R} du
\,\left|\Phi^\lambda_{T-t}(-u-i)\Phi^\lambda_{T-t}(-v-i)
  \Phi^\lambda_t(-u-v-2i)\right|(\mathbf 1_{|\lambda u\leq l|} +
\frac{\mathbf 1_{\lambda u >l}}{\lambda\sqrt{1+u^2}}) ( \mathbf 1_{|\lambda v\leq l|} +
\frac{ \mathbf 1_{\lambda v >l}}{\lambda\sqrt{1+v^2}}).\label{studyint.eq}
\end{align}
We shall decompose it into four terms corresponding to integrals over
non-disjoint sets (whose union is $\mathbb R^2$)
$\{|u\lambda|\leq 2\ell, |v\lambda|\leq 2\ell\}$, $\{|u\lambda|> 2\ell,
|v\lambda|\leq \ell\}$, $\{|u\lambda|\leq \ell, |v\lambda|> 2\ell\}$ and
$\{|u\lambda|> \ell, |v\lambda|> \ell\}$, and show that on the first
set one can apply the dominated convergence theorem, and the
contribution of the three other sets to the limit is zero. On the first set, the integrand is
bounded as follows:
$$
\left|\Phi^\lambda_{T-t}(-u-i)\Phi^\lambda_{T-t}(-v-i)
  \Phi^\lambda_t(-u-v-2i)\right|\mathbf 1_{\{|u\lambda|\leq \ell,
  |v\lambda|\leq \ell\}} \leq e^{-c(T-t)
  (u^2+v^2) -ct (u+v)^2},
$$
which is integrable in $u,v$ and $t$. Hence, by the dominated
convergence theorem,
\begin{align*}
&\frac{\lambda^{-2}}{4\pi^2} \int_{(0,\infty)}\mu(dK) \int_{(0,\infty)}\mu(d\bar K)\int_0^T dt \int_{|v\lambda |\leq 2\ell}
dv\int_{|u\lambda|\leq 2\ell} du \,
K^{iu}\bar K^{iv}\\ &\qquad \qquad \times\frac{\Phi^\lambda_{T-t}(-u-i)\Phi^\lambda_{T-t}(-v-i)
  \Phi^\lambda_t(-u-v-2i)}{-uv(iu-1)(iv-1)}\left\{-\frac{a_\lambda(u)a_\lambda(v)}{\bar
  \sigma^2} + b_\lambda(u,v) \right\} \\ &\to \frac{1}{4}\left(m_4 - \frac{m_3^2}{\bar\sigma^2}\right) \mathbb E^{BS}\left[\int_0^T \left(S_t^2
    \frac{\partial P_{BS}(t,S_t)}{\partial S^2}\right)^2 dt\right]
\end{align*}
as $\lambda\to 0$. It remains to show that the other three sets give a
zero contribution to the limit. 

On the set $\{|u\lambda|> 2\ell,
|v\lambda|\leq \ell\}$, the integrand in \eqref{studyint.eq} is bounded by:
\begin{align*}
&\frac{\lambda^{-1}}{\sqrt{1+u^2}} e^{-c (T-t)\lambda^{\beta-2} |u|^\beta -
  c(T-t)v^2  - ct\lambda^{\beta-2}|u+v|^\beta }\leq \frac{\lambda^{-1}}{\sqrt{1+u^2}} e^{-c (T-t)\lambda^{\beta-2} |u|^\beta -
  c(T-t)v^2  - ct\lambda^{\beta-2}|v|^\beta }\\
&\leq \frac{\lambda^{-1}}{\sqrt{1+u^2}} e^{-c (T-t)\lambda^{\beta-2} |u|^\beta   - cT(\ell^{\beta -2}\wedge 1) |v|^2 }.
\end{align*}
On the other hand,
\begin{align*}
&\int_{|u\lambda|> 2\ell} \frac{du}{\lambda \sqrt{1+u^2}} e^{-c
  (T-t)\lambda^{\beta-2} |u|^\beta} \leq  2\int_{ 2\ell}^\infty
\frac{du}{\lambda u} e^{-c
  (T-t)\lambda^{-2} |u|^\beta} = \frac{1}{\sqrt{T-t}}f\left(\frac{T-t}{\lambda^2}\right),
\end{align*}
where
\begin{align*}
f(\theta) =2\sqrt{\theta} \int_{2\ell \theta^{1/\beta}}^\infty
\frac{du}{u} e^{-c|u|^\beta}
\end{align*} 
Note that $f$ is a bounded positive function and $f(\theta)\to 0$ as
$\theta \to \infty$. Therefore, by the dominated convergence theorem,
$$
\int_0^T dt \int_{|u\lambda|> 2\ell} du \int_{|v\lambda|\leq \ell}  \left|\Phi^\lambda_{T-t}(-u-i)\Phi^\lambda_{T-t}(-v-i)
  \Phi^\lambda_t(-u-v-2i)\right|\frac{1}{\lambda \sqrt{1+u^2}} \to 0
$$
as $\lambda \to 0$. 

The set $\{|u\lambda|\leq \ell,
|v\lambda|> 2\ell\}$ can be dealt with in the same manner. Finally, on
the set $\{|u\lambda|>\ell,
|v\lambda|> \ell\}$, the integrand in \eqref{studyint.eq} is bounded
by 
\begin{multline}
\frac{\lambda^{-2}}{\sqrt{(1+u^2)(1+v^2)}} \Big\{e^{-c (T-t)\lambda^{\beta-2} |u|^\beta -
  c(T-t) \lambda^{\beta-2} |v|^\beta -ct |u+v|^2}\mathbf 1_{\lambda|u+v|\leq\ell } \\+e^{-c (T-t)\lambda^{\beta-2} |u|^\beta -
  c(T-t) \lambda^{\beta-2} |v|^\beta  - ct \lambda^{\beta-2}|u+v|^\beta } \Big\}\label{largeuv}
\end{multline}
With a change of variable $u = \frac{x+y}{2}$, $v = \frac{x-y}{2}$,
using the convexity inequality
$$
\left|\frac{x+y}{2}\right|^\beta+\left|\frac{x-y}{2}\right|^\beta\geq
c_\beta(|x|^\beta + |y|^\beta),
$$ 
the integral of the first term above satisfies
\begin{align*}
&\int_{|u\lambda|\geq \ell}du \int_{|v\lambda|\geq \ell}dv \frac{\lambda^{-2}}{\sqrt{(1+u^2)(1+v^2)}} e^{-c (T-t)\lambda^{\beta-2} |u|^\beta -
  c(T-t) \lambda^{\beta-2} |v|^\beta -ct |u+v|^2}\mathbf
1_{\lambda|u+v|\leq\ell } \\
&\leq \int_{\lambda|x+y|>2\ell }dx \int_{\lambda|x-y|>2\ell}dy
\frac{\lambda^{-2}}{(1+|x+y|)(1+|x-y|)} e^{-c c_\beta (T-t)\lambda^{\beta-2} |x|^\beta -
  cc_\beta(T-t) \lambda^{\beta-2} |y|^\beta -ct |x|^2}\mathbf
1_{\lambda|x|\leq\ell } \\
&\leq \int_{\lambda|y|>\ell }dy \int_{\lambda|x|>\ell}dx
\frac{\lambda^{-2}}{(1+|x+y|)(1+|x-y|)} e^{-c(c_\beta\wedge 1)T |x|^2}
\leq \frac{C}{\lambda^2}\int_{|x|>\ell/\lambda}  dx \,e^{-c(c_\beta\wedge 1)T |x|^2},
\end{align*}
for some constant $C<\infty$, where we used Lemma 2 in \cite{broden.tankov.09}.
It is clear that this expression converges to $0$ as $\lambda \to 0$.

Finally, 
using the same change of variable as above, and once again, Lemma 2 in
\cite{broden.tankov.09}, the integral of the second term in \eqref{largeuv} satisfies,
\begin{align*}
&\int_{|u\lambda|\geq \ell}du\int_{|v\lambda|\geq \ell}dv
\frac{\lambda^{-2}}{\sqrt{(1+u^2)(1+v^2)}} e^{-c (T-t)\lambda^{\beta-2} |u|^\beta -
  c(T-t) \lambda^{\beta-2} |v|^\beta  - ct
  \lambda^{\beta-2}|u+v|^\beta }\\
& \leq   \int_{\lambda|x+y|>2\ell }dx \int_{\lambda|x-y|>2\ell}dy
\frac{\lambda^{-2}}{(1+|x+y|)(1+|x-y|)} e^{-c (c_\beta \wedge 1)T\lambda^{\beta-2} |x|^\beta }\\
&\leq \frac{C}{\lambda^2}\int_{\lambda|x|>\ell }dx \,e^{-c (c_\beta
  \wedge 1)T\lambda^{\beta-2} |x|^\beta } = \frac{C}{\lambda^3}\int_{|x|>\ell }dx \,e^{-c (c_\beta
  \wedge 1)T\lambda^{-2} |x|^\beta } = {C
  \lambda^{{2\over\beta} - 3}}\int_{|x|>\ell \lambda^{-{2\over \beta}} }dx \,e^{-c (c_\beta
  \wedge 1)T |x|^\beta }
\end{align*}
which clearly goes to zero as $\lambda \to 0$.

To perform the computation for the put option pay-off, recall that in
the Black-Scholes model, $S_t= S_0 e^{-\frac{\sigma^2}{2}t +\sigma
  W_t}$ and $\frac{\partial^2 \text{Put}^K(t,S)}{\partial S^2} =
\frac{\phi(d_1(t))}{S\sigma\sqrt{T-t}}$ with $\phi(x) =
\frac{e^{-\frac{x^2}{2}}}{\sqrt{2\pi}}$. Therefore,
\begin{align*}
\mathbb{E}^{BS} \left[ \int_0^T \left(S_t^2 \frac{\partial^2 \text{Put}^K(t,S)}{\partial S^2} \right)^2 dt\right]
&= \frac{K^2}{2\pi \sigma^2} \int_0^T \mathbb{E}^{BS}  
\left[ e^{-\left(\frac{\log (S_t/K)}{\sigma\sqrt{T-t}}-\frac{\sigma\sqrt{T-t}}{2} \right)^2} \right] {
\frac{dt}{T-t}}\\
&= \frac{K^2}{2\pi \sigma^2} \int_0^T \mathbb{E}^{BS}  
\left[ e^{-\left(\frac{\log (S_0/K)- \frac{\sigma^2 T}{2} + \sigma W_t}{\sigma\sqrt{T-t}}\right)^2} \right] {
\frac{dt}{T-t}}.
\end{align*}
It remains to perform the explicit integration with the Gaussian density to get
the result.
\section{Proof of Lemma \ref{linear.lm}}
In this proof, 
we denote $X^{BS}_t = \bar \sigma W_t$, where $W$ is a standard Brownian motion
independent from $X^\lambda$. 
For $\lambda>0$, define 
$$f_\lambda(t,x)=
\mathbb{E}[h(x\mathcal{E}(X^\lambda)_t\mathcal{E}(X^{BS})_{T-t})],$$
and for $\varepsilon>0$, let 
\begin{align*}
&h^\varepsilon(x) =
\mathbb{E}[h(x\mathcal{E}(X^{BS})_\varepsilon)],\qquad
P^\varepsilon_{BS}(t,x) = \mathbb E[h^\varepsilon(x \mathcal
E(X^{BS})_{T-t})] = h^{\varepsilon + T-t}(x),\\
 & f_\lambda^\varepsilon(t,x)=
\mathbb{E}[h^\varepsilon(x\mathcal{E}(X^\lambda)_t\mathcal{E}(X^{BS})_{T-t})].
\end{align*}
These functions are well defined because $X^\lambda$ has bounded jumps
and therefore all moments of $\mathcal E(X^\lambda)_t$ are finite. 
Without loss of generality we shall also take $S_0=1$ below.

From It\={o}  formula, using item \ref{itemBS} of Lemma \ref{lemcv}, 
\begin{align*}
P^\varepsilon_{BS}(T,S_T^\lambda) &= P^\varepsilon_{BS}(0,1)
+ \int_0^T \frac{\partial P^\varepsilon_{BS}}{\partial t} dt + \int_0^T \frac{\partial P^\varepsilon_{BS}}{\partial S} S_t^\lambda \sigma dW_t 
+ \int_0^T \frac{1}{2}\sigma^2(S_t^\lambda)^2 \frac{\partial^2 P^\varepsilon_{BS}}{\partial S^2}dt\\
&+ \sum \{P^\varepsilon_{BS}(t,S_{t-}^\lambda(1+z)) - P^\varepsilon_{BS}(t,S_{t-}^\lambda) 
-zS_t^\lambda\frac{\partial P^\varepsilon_{BS}}{\partial S}(t,S_{t}^\lambda)\}  J^\lambda(dsdz).
\end{align*}
Taking the expectation and using the Black-Scholes equation and
Fubini's theorem justified by item \ref{itemfl} of Lemma \ref{lemcv}, we get:
\begin{align*}
&\mathbb{E} [h^\varepsilon(S_T^\lambda)] - P^\varepsilon_{BS}(0,1)
\\ &= \mathbb{E}\int_0^T \int_{\mathbb{R}} 
\left\lbrace P^\varepsilon_{BS}(t,S_t^\lambda(1+z))-P^\varepsilon_{BS}(t,S_t^\lambda) -zS_t^\lambda\frac{\partial P^\varepsilon_{BS}}{\partial S}
-\frac{1}{2}z^2(S_t^\lambda)^2\frac{\partial^2
  P^\varepsilon_{BS}}{\partial S^2} \right\rbrace \nu_\lambda(dz)dt\\
& = \int_0^T \int_{\mathbb{R}} 
\left\lbrace f^\varepsilon_\lambda(t,1+z)-
  f^\varepsilon_\lambda(t,1)-z (\partial_x f_\lambda^\varepsilon)(t,1) -
  \frac{z^2}{2} (\partial^2_x f_\lambda^\varepsilon)(t,1) \right\rbrace
\nu_\lambda(dz)dt\\
& = \int_0^T \int_{\mathbb{R}} \left\{ \frac{z^3}{6} (\partial^3_x
  f_\lambda^\varepsilon)(t,1)  + \int_0^1 \frac{z^4}{6} (1-\theta)^3 (\partial^4_x f_\lambda^\varepsilon)(t,1+z\theta) 
d\theta \right\}\nu_\lambda(dz)dt
\end{align*}

Define $\tilde{f}_\lambda^\varepsilon(t,x) = (\partial^3_x
  f_\lambda^\varepsilon)(t,x)=
\mathbb{E}[\tilde{h}^\varepsilon(x\mathcal{E}(X^\lambda)_t\mathcal{E}(X^{BS})_{T-t})]$,
where $\tilde{h}^\varepsilon(x) = x^3 (\partial_x^3 h^\varepsilon)
(x)$. Then, using again a Taylor-Lagrange expansion, we get  for all $t\in[0,T]$
\begin{align*}
\tilde{f}_\lambda^\varepsilon(t,1)&= \tilde{f}_\lambda^\varepsilon(0,1) +\int_0^t \int_\mathbb{R} \lbrace \tilde{f}_\lambda^\varepsilon(s,1+z)-\tilde{f}_\lambda^\varepsilon(s,1)-z (\partial_x \tilde{f}_\lambda^\varepsilon)(s,1) 
- \frac{z^2}{2} (\partial^2_x \tilde{f}_\lambda^\varepsilon)(s,1) \rbrace \nu_\lambda^\varepsilon(dz) ds\\
&=\frac{\partial^3 P^\varepsilon_{BS}}{\partial S^3}(0,1)
+ \int_0^t \int_\mathbb{R} \int_0^1 \frac{\lambda z^3}{2} (1-\theta)^2 (\partial_x^3 \tilde{f}_\lambda^\varepsilon)(s,1+ \lambda z \theta)  d\theta \nu(dz) ds.
\end{align*}
Substituting this representation into the above formula, we obtain
\begin{align*}
\mathbb{E} [h^\varepsilon(S_T^\lambda)]= P^\varepsilon_{BS}(0,1) &+ \frac{\lambda m_3T}{6} \frac{\partial^3 P^\varepsilon_{BS}}{\partial S^3}(0,1)
+ \lambda^2 \int_0^T  dt \int_\mathbb{R} \frac{z^4}{6} \nu(dz) \int_0^1 (1-\theta)^3  (\partial_x^4 f_\lambda^\varepsilon)(t,1+\lambda z\theta)d\theta \\
&+\frac{\lambda^2 m_3}{6} \int_0^T  dt \int_0^t ds \int_\mathbb{R}\frac{z^3}{2}  \nu(dz) \int_0^1
(1-\theta)^2 (\partial_x^3 \tilde{f}_\lambda^\varepsilon)(s,1+\lambda z \theta)d\theta .
\end{align*}
Note that $\forall (s,x)\in [0,T]\times \mathbb{R}_+$
\[(\partial_x^3 \tilde{f}_\lambda^\varepsilon)(s,x) 
= 6  (\partial_x^3 f_\lambda^\varepsilon)(s,x)
+ 18x  (\partial_x^4 f_\lambda^\varepsilon)(s,x)
+ 9x^2  (\partial_x^5 f_\lambda^\varepsilon)(s,x)
+x^3 (\partial_x^6 f_\lambda^\varepsilon)(s,x).\]
Now we use item \ref{itemeps} of Lemma \ref{lemcv} to make $\varepsilon$ go to zero,
obtaining
\begin{align*}
\mathbb{E} [h(S_T^\lambda)] = P_{BS}(0,1) &+ \frac{\lambda m_3T}{6} \frac{\partial^3 P^\varepsilon_{BS}}{\partial S^3}(0,1)
+ \lambda^2 \int_0^T  dt \int_\mathbb{R} \frac{z^4}{6} \nu(dz) \int_0^1 (1-\theta)^3  (\partial_x^4 f_\lambda)(t,1+\lambda z\theta)d\theta \\
&+\frac{\lambda^2 m_3}{6} \int_0^T  dt \int_0^t ds \int_\mathbb{R}\frac{z^3}{2}  \nu(dz) \int_0^1
(1-\theta)^2 (\partial_x^3 \tilde{f}_\lambda)(s,1+\lambda z \theta)d\theta.
\end{align*}
To finish the proof, we use the dominated convergence theorem
(justified by items \ref{itemfl} and \ref{itemlambda} of Lemma \ref{lemcv}) to show that 
\begin{align*}
\int_0^T  dt \int_\mathbb{R} \frac{z^4}{6} \nu(dz) \int_0^1
(1-\theta)^3  (\partial_x^4 f_\lambda)(t,1+\lambda z\theta)d\theta& \to \int_0^T  dt \int_\mathbb{R} \frac{z^4}{6} \nu(dz) \int_0^1
(1-\theta)^3  (\partial_x^4 f_0)(t,1)d\theta\\
& =\frac{m_4 T}{24} (\partial_x^4 f_0)(0,1)  = \frac{m_4
  T}{24}  { \frac{\partial^4 P_{BS}}{\partial S^4}(0,1)}
\end{align*}
and similarly
\begin{align*}
&\int_0^T  dt \int_0^t ds \int_\mathbb{R}\frac{z^3}{2}  \nu(dz) \int_0^1
(1-\theta)^2 (\partial_x^3 \tilde{f}_\lambda)(s,1+\lambda z
\theta)d\theta \to \frac{m_3 T^2}{12} (\partial_x^3 \tilde f_0)(0,1)  \\
& = \frac{m_3 T^2}{12}\left\{6 P^{(3)}_{BS}(1) + 18 
  P^{(4)}_{BS}(1) + 9 P^{(5)}_{BS}(1) +  P^{(6)}_{BS}(1)\right\}.
\end{align*}

\begin{lem} 
\label{lemcv} 
Let the assumptions of Lemma \ref{linear.lm} hold true. 
Then, for all $0\leq k\leq 6$
\begin{enumerate}
\item $P^\varepsilon_{BS} \in C^{1,2}([0,T]\times (0,\infty))$. \label{itemBS}
\item $(\partial_x^k h^\varepsilon)(x) $ exists, is continuous and has
  polynomial growth in $x$ for all
  $\varepsilon>0$.
\item \label{itemfl} $(\partial_x^k f^\varepsilon_\lambda)(t,x) $ exists, is
  continuous in $x$ and satisfies
$$
|(\partial_x^k f^\varepsilon_\lambda)(t,x)|\leq C(1+ |x|^n)
$$
for some $n\geq 0$ and a constant $C$ which does not depend on $t$,
$\varepsilon$ or $\lambda$. 
\item For all $t,x$, $(\partial_x^k f_\lambda)(t,x) \to (\partial_x^k h^T)(x)$ as $\lambda \to 0$.\label{itemlambda}
\item For all $t$, $x$, $\lambda$, $(\partial_x^k f_\lambda^\varepsilon)(t,x) \to (\partial_x^k f_\lambda)(t,x)$ as $\varepsilon \to 0$.\label{itemeps}
\end{enumerate}
\end{lem}

\begin{proof}
Under our assumptions, the random variable
$$
\log(\mathcal E(X^\lambda)_t \mathcal E(X^{BS})_{T-t+\varepsilon})
$$
admits a density $p_\lambda^\varepsilon(t,x)$ which can be recovered
via Fourier inversion:
$$
p_\lambda^\varepsilon(t,x) = \frac{1}{2\pi}\int_{\mathbb R} e^{-iux}
\Phi^{BS}_{T-t+\varepsilon}(u) \Phi^\lambda_t(u)du. 
$$
By the bound \eqref{boundlambda.eq} and the explicit form of
$\Phi^{BS}$, we conclude that the derivatives of
$p_\lambda^\varepsilon(t,x)$ with respect to $x$ of any order are
continuous and given by 
$$
\partial^k_x p_\lambda^\varepsilon(t,x) = \frac{1}{2\pi}\int_{\mathbb
  R} (-iu)^k e^{-iux}
\Phi^{BS}_{T-t+\varepsilon}(u) \Phi^\lambda_t(u)du.
$$
By Jensen's inequality and Plancherel's theorem, for any $p\geq 0$, 
\begin{align*}
&\int_{\mathbb R}|\partial^k_x p_\lambda^\varepsilon(t,x)| e^{p|x|} dx
= \int_{\mathbb R}|\partial^k_x p_\lambda^\varepsilon(t,x)
e^{(p+2)|x|}|e^{-|x|} dx \leq   \left(\int_{\mathbb R}|\partial^k_x p_\lambda^\varepsilon(t,x)
|^2 e^{(2p+2)|x|}\right)^{\frac{1}{2}} \\ &\leq \left(\int_{\mathbb R}|\partial^k_x p_\lambda^\varepsilon(t,x)
|^2 e^{(2p+2)x}dx\right)^{\frac{1}{2}}+\left(\int_{\mathbb R}|\partial^k_x p_\lambda^\varepsilon(t,x)
|^2 e^{-(2p+2)x}dx\right)^{\frac{1}{2}}\\
& = \left(\frac{1}{2\pi}\int_{\mathbb R}|
v-i(p+1)|^{2k} |\Phi^{BS}_{T-t+\varepsilon}(v-i(p+1))
\Phi^\lambda_t(v-i(p+1))|^2 dv\right)^{\frac{1}{2}}\\ 
& + \left(\frac{1}{2\pi}\int_{\mathbb R}|
v+i(p+1)|^{2k} |\Phi^{BS}_{T-t+\varepsilon}(v+i(p+1))
\Phi^\lambda_t(v+i(p+1))|^2 dv\right)^{\frac{1}{2}}<\infty. 
\end{align*}
Consider for example, the second term. Using the bound
\eqref{boundlambda.eq}, it satisfies, for some constant $C<\infty$,
\begin{align}
&\int_{\mathbb R}|
v+i(p+1)|^{2k} |\Phi^{BS}_{T-t+\varepsilon}(v+i(p+1))
\Phi^\lambda_t(v+i(p+1))|^2 dv\notag\\ &\leq C\int_{\mathbb R}(1+|v|^{2k})
e^{-(T-t+\varepsilon)\bar \sigma^2 v^2 - ct |v|^2 \mathbf
  1_{|v\lambda|\leq \ell} - ct |v|^\beta \lambda^{\beta-2}
  \mathbf 1_{|v\lambda|>\ell}} dv\notag\\
&\leq C\int_{\mathbb R}(1+|v|^{2k})
e^{-(T-t)\bar \sigma^2 v^2}(e^{- ct |v|^2} + e^{- ct |v|^\beta}) dv,\label{intbound}
\end{align}
which is easily seen to be bounded uniformly on $t$. Therefore, 
$\int_{\mathbb R}|\partial^k_x p_\lambda^\varepsilon(t,x)| e^{p|x|}
dx$ is bounded uniformly on $t$, $\varepsilon$ and $\lambda$. 

This means that the function $f^\varepsilon_\lambda$ is given by
$$
f^\varepsilon_\lambda(t,x) = \int_{\mathbb R} dz\, h(xe^z)
p^\varepsilon_\lambda (t,z).
$$
Instead of the function $f^\varepsilon_\lambda$ we shall, for notational
convenience, study the function $\bar f^\varepsilon_\lambda(t,x) =
f^\varepsilon_\lambda(t,e^x)$, which is therefore given by
$$
\bar f^\varepsilon_\lambda(t,x) = \int_{\mathbb R} dz\, h(e^{z+x})
p^\varepsilon_\lambda (t,z) = \int_{\mathbb R} dz\, h(e^{z})
p^\varepsilon_\lambda (t,z-x) 
$$
By dominated convergence, using the above estimate, we then get that 
$$
\partial^k_x \bar f^\varepsilon_\lambda(t,x) = (-1)^k\int_{\mathbb R} dz\,
h(e^z) \partial^k_x  p^\varepsilon_\lambda(t,z-x) = (-1)^k\int_{\mathbb R} dz\,
h(e^{x+z}) \partial^k_x  p^\varepsilon_\lambda(t,z) 
$$
exists, is continuous and has exponential growth in $x$, which means
that $\partial^k_x  f^\varepsilon_\lambda(t,x)$ has polynomial
growth uniformly on  {  $t$, $\varepsilon$ and $\lambda$}. This finishes the
proof of item \ref{itemfl}. 

To study the convergence in $\lambda$, remark that from the polynomial
growth of $h$, $|h(e^x)| \leq
e^{p|x|}$. Then, proceeding similarly to the above, we have
\begin{align*}
&|\partial^k_x \bar f_\lambda(t,x) - \partial^k_x \bar f_0(t,x)|\leq
C \int _{\mathbb R}e^{p|z|} |\partial^k_x  p_\lambda(t,z)
- \partial^k_x  p_0(t,z) |\\
&\leq C \left(\int_{\mathbb R}(1+|
v|^{2k}) |\Phi^{BS}_{T-t}(v-i(p+1))
(\Phi^\lambda_t(v-i(p+1))-\Phi^{BS}_t(v-i(p+1))|^2 dv\right)^{\frac{1}{2}}\\ 
& + C \left(\int_{\mathbb R}(1+|
v|^{2k}) |\Phi^{BS}_{T-t}(v+i(p+1))
(\Phi^\lambda_t(v+i(p+1))-\Phi^{BS}_t(v+i(p+1)))|^2 dv\right)^{\frac{1}{2}}.
\end{align*}
Consider for example the second term. It satisfies
\begin{align*}
&\int_{\mathbb R}(1+|
v|^{2k}) |\Phi^{BS}_{T-t}(v+i(p+1))
(\Phi^\lambda_t(v+i(p+1))-\Phi^{BS}_t(v+i(p+1)))|^2 dv \\
&\leq C \int_{\mathbb R}(1+|
v|^{2k}) e^{-(T-t)\bar \sigma^2 v ^2 }|\Phi^\lambda_t(v+i(p+1))-\Phi^{BS}_t(v+i(p+1)))|^2 dv
\end{align*}
Since $\Phi^\lambda_t(v+i(p+1))\to \Phi^{BS}_t(v+i(p+1))$ for all $v$
as $\lambda \to 0$, and an integrable bound can be found similarly to
\eqref{intbound}, we conclude using the dominated convergence theorem
that the above expression converges to $0$ as $\lambda \to 0$. This
finishes the proof of item \ref{itemlambda}. Other items can be proved
in a similar manner. 

\end{proof}


\section{A general formula for high-order Black-Scholes greeks}\label{greeks.app}

The risk neutral price process dynamics of the stock in the 
Black-Scholes model with zero interest rate and volatility $\sigma>0$ reads
\[dS_t = \sigma S_t dW_t, 
\]
where $W$ denotes a standard Brownian motion. In this model, the pricing function of a European call option with strike $K>0$ and maturity $T>0$ given by $P :  [0,T]\times\mathbb{R}_+ \to \mathbb{R}_+$ satisfies
\[P(t,s)=\mathbb{E}[(S_T-K)_{{+}}|S_t=s]= s\Phi(\delta_1(t,s)) - K\Phi(\delta_2(t,s)),
\]
where $\Phi$ and $\varphi$ are, respectively, the cumulative distribution function
and the density of the standard normal distribution, and the
coefficients $\delta_1$ and $\delta_2$ are defined by 
$$
\delta_1(t,s)=\frac{\log \frac{s}{K}+\frac{\sigma^2}{2}(T-t)}{\sigma\sqrt{T-t}},\qquad
\delta_2(t,s) = \delta_1(t,s) - \sigma \sqrt{T-t}.
$$
 We set for $n\in \mathbb{N}$,
\[d_n(t,s)= s^n\frac{\partial^n P}{\partial s^n}(t,s).
\]
The first two cash greeks can be computed by direct differentiation:
\begin{align*}
&d_1(t,s) =s\Phi (\delta_1(t,s))\\
&d_2(t,s)= s^2\frac{\varphi(\delta_1(t,s))}{s\sigma\sqrt{T-t}}.
\end{align*}
For higher order derivatives {of European call/put option prices}, the following recurrence relation holds
for all $ n\geq 0$:
\[d_{3+n}(t,s)= \sum_{k=0}^n C_k^n D_{n-k}(t,s) d_{2+k}(t,s) \]
where $C_k^n = \binom{n}{k} = \frac{n!}{k!(n-k)!}$ are the binomial
coefficients,
$$D_k(t,s)=(-1)^{k+1}k! \left[ \delta(t,s) - \frac{1}{\sigma^2
    (T-t)}\sum_{p=1}^k \frac{1}{p} \right]\quad \text{and}\quad \delta(t,s) = \frac{\delta_1(t,s)}{\sigma \sqrt{T-t}} +1.$$ 



This recurrence relation leads to the following {formulae} for the cash
greeks up to order $6$:
\begin{align*}
&d_3(t,s)= -d_2(t,s)\delta(t,s),\\
&d_4(t,s)= d_2(t,s)\left(\delta(t,s)-\frac{1}{\sigma^2 \tau} \right) -d_3(t,s)\delta(t,s),\\
&d_5(t,s)=-d_2(t,s)\left(2\delta(t,s) - \frac{3}{\sigma^2\tau}\right) + 2d_3(t,s) \left( \delta(t,s)-\frac{1}{\sigma^2\tau} \right)-d_4(t,s)\delta(t,s),\\
&d_6(t,s)= d_2(t,s)\left(6\delta(t,s)-\frac{11}{\sigma^2\tau}\right)
-3d_3(t,s)\left(2\delta(t,s)-\frac{3}{\sigma^2\tau}\right) \\ &\qquad \quad+3d_4(t,s)\left(\delta(t,s)-\frac{1}{\sigma^2\tau}\right) -d_5(t,s)\delta(t,s).
\end{align*}

\end{document}